\documentclass[aps,superscriptaddress,twocolumn,12pt,twoside,floatfix,pra,a4paper,nofootinbib]{revtex4-2} 

\usepackage{makecell}
\usepackage{array}
\newcolumntype{C}[1]{>{\centering\arraybackslash}m{#1}}
\usepackage{booktabs,adjustbox}
\usepackage{float}
\usepackage{amsmath}
\usepackage[utf8]{inputenc}  
\usepackage[T1]{fontenc}     
\usepackage[table]{xcolor}
\usepackage{color,soul}
\usepackage[british]{babel}  
\usepackage[sc,osf]{mathpazo}
\usepackage[scaled=0.86]{berasans}  
\usepackage[colorlinks=true, citecolor=blue, urlcolor=blue]{hyperref}  
\usepackage{graphicx} 
\usepackage[babel]{microtype}  
\usepackage{amsmath,amssymb,amsthm,bm,amsfonts,mathrsfs,bbm} 

\usepackage{xspace}  
\usepackage{pgf,tikz}
\usepackage{xcolor}
\usepackage{multirow}
\usepackage{array}
\usepackage{bigstrut}
\usepackage{braket}
\usepackage{color}
\usepackage{natbib}
\usepackage{multirow}
\usepackage{mathtools}
\usepackage{float}
\usepackage{blkarray}
\usepackage[caption = false]{subfig}
\usepackage{xcolor,colortbl}
\usepackage{color}
\usepackage{rotating}
\usepackage{tikz}
\usetikzlibrary{quantikz}
\usepackage{tabularx}
\usepackage{relsize}

\newcommand{\be}{\begin{equation}}
\newcommand{\ee}{\end{equation}}
\newcommand{\ba}{\begin{eqnarray}}
\newcommand{\ea}{\end{eqnarray}}

\newtheorem{definition}{Definition}
\newtheorem{proposition}{Proposition}
\newtheorem{observation}{Observation}

\newtheorem{Lemma}{Lemma}
\newtheorem{remark}{Remark}

\def\>{\rangle}
\def\<{\langle}

\usepackage{centernot}
\usepackage{subfig}

\begin{document}

\title{Overcoming Traditional No-Go Theorems: Quantum Advantage in Multiple Access Channels}

\author{Ananya Chakraborty}
\affiliation{Department of Physics of Complex Systems, S. N. Bose National Center for Basic Sciences, Block JD, Sector III, Salt Lake, Kolkata 700106, India.}

\author{Sahil Gopalkrishna Naik}
\affiliation{Department of Physics of Complex Systems, S. N. Bose National Center for Basic Sciences, Block JD, Sector III, Salt Lake, Kolkata 700106, India.}

\author{Edwin Peter Lobo}
\affiliation{Laboratoire d’Information Quantique, Universite libre de Bruxelles (ULB), Av. F. D. Roosevelt 50, 1050 Bruxelles, Belgium.}

\author{Ram Krishna Patra}
\affiliation{Department of Physics of Complex Systems, S. N. Bose National Center for Basic Sciences, Block JD, Sector III, Salt Lake, Kolkata 700106, India.}

\author{Samrat Sen}
\affiliation{Department of Physics of Complex Systems, S. N. Bose National Center for Basic Sciences, Block JD, Sector III, Salt Lake, Kolkata 700106, India.}

\author{Mir Alimuddin}
\affiliation{Department of Physics of Complex Systems, S. N. Bose National Center for Basic Sciences, Block JD, Sector III, Salt Lake, Kolkata 700106, India.}

\author{Amit Mukherjee}
\affiliation{Indian Institute of Technology  Jodhpur, Jodhpur 342030, India.}

\author{Manik Banik}
\affiliation{Department of Physics of Complex Systems, S. N. Bose National Center for Basic Sciences, Block JD, Sector III, Salt Lake, Kolkata 700106, India.}

\begin{abstract}
Extension of point-to-point communication model to the realm of multi-node configurations finds a plethora of applications in internet and telecommunication networks. Here, we establish a novel advantage of quantum communication in a commonly encountered network configuration known as the Multiple Access Channel (MAC). A MAC consists of multiple distant senders aiming to send their respective messages to a common receiver. Unlike the quantum superdense coding protocol, the advantage reported here is realized without invoking entanglement between the senders and the receiver. Notably, such an advantage is unattainable in traditional point-to-point communication involving one sender and one receiver, where the limitations imposed by the Holevo and Frankel  Weiner no-go theorems come into play. Within the MAC setup, this distinctive advantage materializes through the receiver's unique ability to simultaneously decode the quantum systems received from multiple senders. Intriguingly, some of our MAC designs draw inspiration from various other constructs in quantum foundations, such as the Pusey-Barrett-Rudolph theorem and the concept of `nonlocality without entanglement', originally explored for entirely different purposes. Beyond its immediate applications in network communication, the presented quantum advantage hints at a profound connection with the concept of `quantum nonlocality without inputs' and holds the potential for semi-device-independent certification of entangled measurements.
\end{abstract}



\maketitle
\onecolumngrid
\section{Introduction}
The elementary model of communication, originally formulated in Shannon's 1948 seminal work \cite{Shannon1948}, deals with reliable transmission of information between two distant servers. Its quantum version -- the Quantum Shannon theory -- aims to devise advanced methods for information transmission by harnessing the non-classical properties of quantum systems \cite{Wilde2011}. For instance, the groundbreaking quantum superdense coding protocol efficiently employs the intriguing concept of quantum entanglement to transmit two bits of classical message by communicating a two-label quantum system -- a qubit \cite{Bennett1992}. Quantum advantages, however, come along with their own limitations. Holevo's theorem is one such no-go result that limits the capacity of a quantum system to be same as its classical counterpart when no entanglement is allowed between the sender and the receiver \cite{Holevo1973}. Recently, Frenkel and Weiner further generalized this result by demonstrating that in assistance with only pre-shared classical correlation ({\it i.e.} without entanglement) any input/output correlation achievable with an $n$-level quantum system can also be obtained using an $n$-state classical system \cite{Frenkel2015}. This in turn renders the `signaling dimension' of quantum and classical systems to be identical \cite{DallArno2017} (also see \cite{Naik2022,Sen2022,Patra2023}).

Here we report a novel advantage of quantum communication within the network setup which generally involves several distant parties exchanging information among themselves \cite{Gamel2011}. Importantly, the advantage is possible without invoking any pre-shared entanglement between the senders and the receivers. Specifically, we show that communicating a quantum system can be advantageous over its classical counterpart in a widely used network setup called the Multiple Access Channel (MAC) \cite{Ahlswede1971,Liao1972,Biglieri2007,Yard2008}. A MAC consists of multiple distant senders intending to transmit their individual messages to a single receiver, such as uplink from multiple distant mobile phones to a common server (see Fig.\ref{fig1}). In simulation of a given MAC, the aim is to replicate its action using the least possible communication from the senders to the receiver. Interestingly, we showcase several instances of MAC simulation tasks achievable by communicating a qubit from each sender to the receiver, while such simulations become unfeasible if the qubit channels are replaced by classical bit (c-bit) channels. The classical channels can further be assisted with additional side resources, such as classical correlations shared among the parties, also known as shared randomness. Importantly, the reported quantum advantage persists even when the classical communication lines are supplemented with unlimited classical correlations. Notably, some of our MAC designs are inspired from well-established constructs in quantum foundations and quantum information theory, originally studied for entirely different purposes. For instance, drawing motivation from the renowned Pusey-Barrett-Rudolph theorem in quantum foundations \cite{Pusey2012}, we construct a MAC involving two senders with the desired quantum advantage. Additionally, by utilizing a concept from the seminal work on `nonlocality without entanglement' \cite{Bennett1999}, we develop a three-sender MAC which is fundamentally distinct from the earlier one. Apart from these, we also construct  few other MACs, all demonstrating quantum advantage. The origin of the quantum advantage lies in receiver's capability to jointly address the quantum systems received from various senders, which in turn allows them to implement a global decoding. This effectively surpasses the constraints imposed by the Holevo-Frenkel-Weiner no-go results applicable to one-sender-one-receiver communication setup \cite{Holevo1973,Frenkel2015}. While qubit communication leads to the desired advantages without invoking pre-shared entanglement between the senders and receiver, we show that quantum entanglement has nontrivial usages here too. In particular, the reported advantage can be attained with c-bit channels from the senders to the receiver provided they are augmented with pre-shared entanglement. Our findings, coupled with their foundational underpinnings, indicate the potential for various other quantum benefits in network communication and offer the prospect of semi-device-independent certification of non-product measurements which often finds applications in a variety of quantum protocols.

{\it Multiple Access Channel.--} Mathematically, a MAC can be represented as a stochastic map from the senders' input alphabet sets to the receiver's output alphabet set. For instance, a MAC, with $K$ distant senders $\{S_i\}_{i=1}^K$ and one receiver $R$, is described by the stochastic matrix $\mathcal{N}^{K}\equiv\{p(a|x_1,\cdots,x_K)~|~a\in A,~x_i\in X_i\}$, where $X_i$ is the message/input set of the $i^\text{th}$ sender $S_i$, $A$ is the output set of the receiver $R$, and $p(a|x_1,\cdots,x_K)$ denotes the probability of obtaining the outcome $a\in A$ given the inputs $x_i\in X_i$; clearly $\sum_{a\in A} p(a|x_1,\cdots,x_K)=1,~\forall~\Vec{x}\in\times_iX_i$. Consider a scenario where each of the senders can communicate only $1$-bit of classical information to the receiver for simulating a given MAC. Without the assistance of any kind of local or shared randomness, the parties can employ only classical deterministic strategies.
\begin{definition}\label{def1}
A classical deterministic strategy with 1-bit communication from each of the senders to the receiver is an ordered tuple $\left(\mathrm{E}_1,\cdots,\mathrm{E}_K,\mathrm{D}\right)\in\times_{i=1}^K\mathcal{E}_i\times\mathcal{D}$, where $E_i$ denotes a deterministic encoding of $i^{th}$ party's message set $X_i$ into $1$-bit, {\it i.e.}, $\mathrm{E}_i:X_i\mapsto\{0,1\}$ for $i\in\{1,\cdots,K\}$, and $D$ denotes a deterministic function from $K$-bit communication string into the output set $A$, {\it i.e.}, $\mathrm{D}:\{0,1\}^{\times K}\mapsto A$.
\end{definition}
Calligraphic letters in Definition \ref{def1} symbolize sets of all possible deterministic encodings and decodings for the respective parties. Number of such strategies are finite whenever $X_i$'s and $A$ are of finite cardinalities, and collection of such strategies will be denoted as $\mathbf{C}_{ds}^{K}$. Each party, randomizing their respective deterministic strategies locally, can implement a classical local strategy which is an ordered tuple  $\left(P(\mathcal{E}_1),\cdots,P(\mathcal{E}_K),P(\mathcal{D})\right)$ of probability distributions; $P(\mathcal{E}_i)$'s are distributions over encoding functions and $P(\mathcal{D})$ over decoding functions. The set of all such strategies, denoted as $\mathbf{C}_{ls}^{K}$, forms a non-convex set. Shared randomness allowed among the parties empowers them to employ  correlated classical strategies. 
\begin{definition}\label{def2}
A correlated classical strategy utilizes classical shared randomness and is a probability distribution $P(\mathcal{E}_1\times\cdots\times\mathcal{E}_K\times\mathcal{D})$ over the deterministic strategies. \end{definition}
For finite input-output cases, collection of all such strategies, denoted as $\mathbf{C}_{cs}^{K}$, forms a polytope embedded in some $\mathbb{R}^n$; the exact value of $n$ depends on the cardinalities of the input and output sets. The extreme points of this polytope are the deterministic strategies. In actuality, different physical configurations arise depending on what kind of shared randomness (SR) is available. In Definition \ref{def2}, SR is allowed in all possible subgroups of the parties, and the SR resource in this configuration will be denoted as $\$_{G}$. On the other hand, $\cup_{i=1}^K\$_{RS_i}$ denotes a configuration where the receiver shares classical correlation with each of the senders independently (see Fig.\ref{fig1}). If not mentioned otherwise, it will be assumed that unlimited amount of SR is allowed among the parties in an indicated subgroup. Strategies with SR allowed among only some subgroups of parties form non-convex sets lying strictly in between $\mathbf{C}_{ls}^{K}$ and $\mathbf{C}_{cs}^{K}$. With qubit communication the parties can employ a quantum deterministic strategy.
\begin{figure}[t!]
\centering
\includegraphics[scale=0.55]{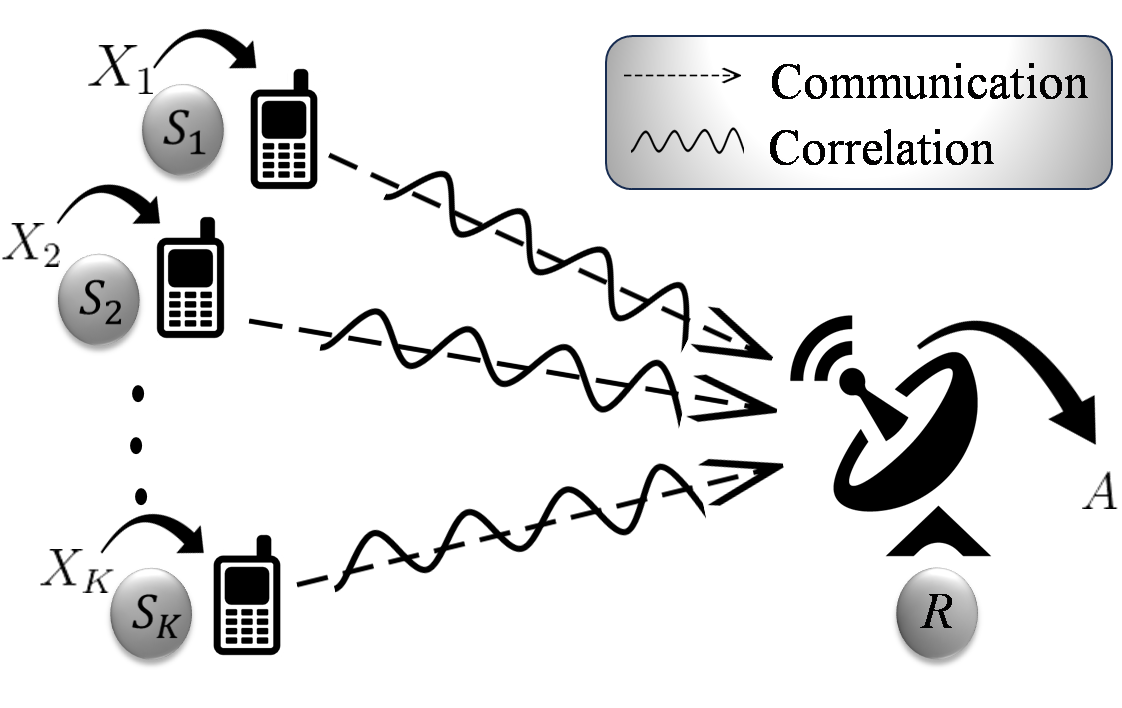}
\caption{Multiple access channel with $K$ distant senders and one receiver, $\mathcal{N}^{K}\equiv\{p(a|x_1,\cdots,x_K)~|~a\in A,~x_i\in X_i\}$. In the configuration, depicted here, each of the senders share correlation with the receiver, but otherwise neither any communication line nor any correlation is available among the senders \cite{self1}.}
\label{fig1}\vspace{-.5cm}
\end{figure}
\begin{definition}\label{def3}
A quantum deterministic strategy with 1-qubit communication from each of the senders to the receiver is an ordered tuple $\left(\mathrm{E}^q_1,\cdots,\mathrm{E}^q_K,\mathrm{D}^q\right)$, where, $\mathrm{E}^q_i:x_i\mapsto\ket{\psi^{x_i}}_{S_i}\in\mathbb{C}^2_{S_i},~\forall~x_i\in X_i$ is an encoding strategy for the $i^{th}$ sender and $\mathrm{D}^q\equiv\left\{\Pi_l\in\mathbb{P}(\otimes_{i=1}^K\mathbb{C}^2_{S_i})~s.t.~\sum_{l=1}^{|A|}\Pi_l=\mathbf{I}\right\}$ is a $|A|$ outcome positive operator valued measure (POVM) performed by the receivers for decoding.
\end{definition}
$\mathbb{P}(\mathcal{H})$ in Definition \ref{def3} denotes the set of positive operators acting on the Hilbert space $\mathcal{H}$. Set of quantum deterministic strategies will be denoted as $\mathbf{Q}_{ds}^{K}$. For the case involving only one sender (with input set denoted as $X$), the Frenkel-Weiner result implies $\mathbf{Q}_{cs}^{1}=\mathbf{C}_{cs}^{1}$, for every input cardinality $|X|$ and every output cardinality $|A|$ \cite{Frenkel2015}, and thus prohibits any advantage of qubit communication over the c-bit.

\section{Quantum Advantages} Going beyond the point-to-point communication scenario, our first example deals with a MAC which involves two senders and one receiver. Each of the senders are given independent two-bit strings ${\bf x}\in\{0,1\}^{\times2}$ and ${\bf y}\in\{0,1\}^{\times2}$, respectively; and the receiver produces a two-bit output strings ${\bf a}\in\{0,1\}^{\times2}$. Quantum strategy reproducing the MAC is as follows: the senders respectively employ the encodings
\begin{align}
\left\{\!\begin{aligned}
\mathrm{E}^q_1:00\mapsto\ket{0},~01\mapsto\ket{+},~10\mapsto\ket{-},~11\mapsto\ket{1}\\
\mathrm{E}^q_2:00\mapsto\ket{0},~01\mapsto\ket{-},~10\mapsto\ket{+},~11\mapsto\ket{1}
\end{aligned}\right\};
\end{align}
where $\{\ket{0},\ket{1}\}$ is the computational basis of $\mathbb{C}^2$, and $\ket{\pm}:=(\ket{0}\pm\ket{1})/\sqrt{2}$. The Receiver, on the qubits received from the senders, performs a two-qubit maximally entangled basis measurement 
\begin{align*}
\mathcal{M}:=\left\{\ket{\psi_{ij}}\bra{\psi_{ij}};~\ket{\psi_{ij}}:=(\mathbb{I}\otimes H^{1-j}X^{\overline{i\oplus j}}Z^{i\oplus j})\ket{\phi^+}\right\} 
\end{align*}
and decodes the outcome as `$i~(\overline{i\oplus j})$' when the projector $\ket{\psi_{ij}}\bra{\psi_{ij}}$ clicks. Here, $\ket{\phi^+}:=(\ket{00}+\ket{11})/\sqrt{2}$, $H:\ket{0}~(\ket{1})\to\ket{+}~(\ket{-})$, and $X,Z$ are the Pauli gates. The resulting MAC can be compactly represented as, $\mathcal{N}^{2}_{PBR}\equiv\left\{p({\bf a}|{\bf x},{\bf y})~|~{\bf a},{\bf x},{\bf y}\in\{0,1\}^{\times2}\right\}$, where
\begin{align}
p({\bf a}|{\bf x},{\bf y})=\begin{cases}
1/2,~~\mbox{when}~~{\bf a}={\bf x}\oplus {\bf y};\\
0,~~\mbox{when}~~{\bf a}=\overline{{\bf x}\oplus {\bf y}};\\
1/4,~~\mbox{otherwise};
\end{cases}
\end{align}
${\bf a}={\bf x}\oplus {\bf y}$ denotes bit-wise XOR, {\it i.e.}, $a_i=x_i\oplus y_i$ for $i=1,2$ (explicit form of the stochastic matrix is provided in Appendix \ref{appa}). Notably, this quantum strategy draws inspiration from the renowned Pusey-Barrett-Rudolph (PBR) theorem in quantum foundations \cite{Pusey2012}, which in turn suggests the name $\mathcal{N}^{2}_{PBR}$ for the resulting MAC. By construction, $\mathcal{N}^{2}_{PBR}$ allows a simulation strategy in $\mathbf{Q}^{2}_{ds}$. We now proceed to establish an impossibility result regarding simulation of this MAC with qubit communication replaced by its classical counterpart.   
\begin{proposition}\label{prop1}
$\mathcal{N}^{2}_{PBR}$ cannot be simulated with $1$-bit communication from each sender to the receiver, even when the communication lines are augmented with the resource $\cup_{i=1}^2\$_{RS_2}$. 
\end{proposition}
\begin{proof}
(Outline) Note that few of the conditional probabilities in $\mathcal{N}^{2}_{PBR}\equiv\{p({\bf a}|{\bf x},{\bf y})\}$ are zero. We first identify the classical deterministic strategies that satisfy these zero conditions. As it turns out, only $48$ deterministic strategies in $\mathbf{C}_{ds}^{2}$ are compatible with these zero requirements. Then we show that any strategy obtained through convex mixing of theses $48$ deterministic strategies and reproducing $\mathcal{N}^{2}_{PBR}$ demands all the three parties to share global randomness $\$_{^G}$ among themselves. This completes the proof of our claim, with detailed calculations provided in Appendix \ref{appa}.
\end{proof}

Proposition \ref{prop1} highlights the advantage of qubit communication over the c-bit in a network communication setup. Notably, the quantum advantage is limited in a sense. Although the c-bit channels augmented with the side resource $\cup_{i=1}^2\$_{RS_i}$ cannot simulate $\mathcal{N}^{2}_{PBR}$, a classical strategy is possible if the resource $\$_G$ ({\it i.e.} global SR among the three parties) is allowed. We now introduce a class of two-sender MACs which exhibit a stronger quantum advantage -- classical strategies become impossible even with the side resource $\$_G$. The senders get inputs $x$ and $y$ from the set $\{0, \cdots, m-1\}$, while the receiver produces a binary output $a \in\{0,1\}$. The probabilities $\{p(a=0|x,y)\}$ uniquely determine the MAC since the other values are fixed by normalization. Denoting $p(a=0|x,y)$ as $p(x,y)$, the MAC $\mathcal{N}^{2}_{m}\equiv\{p(x,y)\}$ is defined as 
\begin{align}\label{poly}
&p(x,y):=\mathrm{Tr}\left[\ket{\phi^+}\bra{\phi^+} \rho_x \otimes \rho_y\right],
\end{align}
where $\rho_u= \frac{1}{2} \left(\mathbf{I}+\cos(2u/m) \sigma_Z+\sin(2u/m)\sigma_X\right)$ for $u\in\{x,y\}$. By construction, all these MACs can be simulated by a quantum strategy with qubit encoding. As the encoding states lie on the vertices of an $m$-sided polygon in the $xz$-plane of the Bloch sphere, we refer to them as the polygon-MACs. Our next result establishes a stronger quantum advantage in simulating a family of these polygon-MACs.
\begin{proposition}\label{prop2}
For $m\in\{5,\cdots,9\}$, the polygon-MACs $\mathcal{N}^{2}_{m}$ cannot be simulated using the strategies $\mathbf{C}^{2}_{cs}$.\end{proposition}
\begin{proof}
Analogous to Bell-type inequalities for  the space-like separated scenario, here we construct linear inequalities for each $m\in\{5,\cdots,9\}$. For a given strategy $\mathbf{p}_c\in\mathbf{C}^{2}_{cs}$ the inequality read as
\begin{equation}\label{wm}
\mathbf{w}_m[\mathbf{p}_c]:=\mathbf{w}_m\cdot\mathbf{p}_c:= \sum_{x,y} w_m^{x,y}p_c(x,y) \le k_m,
\end{equation}
where the constant $k_m$ represents the optimal classical bound obtained using c-bit communication from the senders to the receiver. The value of $k_m$ can be found by computing the value of $\mathbf{w}_m[\mathbf{p}_c]$ on all the deterministic stochastic matrices. Representing $\mathbf{w}_m\equiv(w_m^{x,y})$ as a matrix, with $w_m^{x,y}$ denoting the element of the $x^{th}$ row and $y^{th}$ column, the explicit forms  of the witness operators are provided in Appendix \ref{appb}. The quantum strategy of Eq.(\ref{poly}) yielding a higher value than $k_m$ establishes the claim of the proposition. Optimal classical bounds and the corresponding quantum violations of these inequalities are listed in Table \ref{tabpoly}.
\end{proof}
\begin {center}
\begin {table}[t!]
\begin {tabular} {| c || c | c |}
\hline
 ~~ ~$\mathcal {N}^{2}_{m} $ ~~ ~ &  ~~ ~
     Classical Bound  ~~ ~ &  ~~ ~Quantum Value  ~~ ~\\\hline
$m = 5 $ & ~~ $\mathbf {w} _ 5[\mathbf {p} _c]\le4$ ~~ & $15 (\sqrt
{5} - 1)/4 $ \\
\hline
$m = 6 $ & ~~ $\mathbf {w} _ 6[\mathbf {p} _c]\le6$ ~~ & $6 .75 $ \\
\hline
$m = 7 $ & ~~ $\mathbf {w} _ 7[\mathbf {p} _c]\le2$ ~~ & $\approx 2.61443 $\\
\hline
$m = 8 $ & ~~ $\mathbf {w} _ 8[\mathbf {p} _c]\le8$ ~~ & $12\sqrt {2} \
- 8 $ \\
\hline
$m = 9 $ & ~~ $\mathbf {w} _ 9[\mathbf {p} _c]\le2$ ~~ & $\approx 2.76003$ \\
\hline
\end {tabular}
\captionsetup{justification=centering}
\caption {Quantum advantage in simulating polygon-MAC.}\label{tabpoly}
\end {table}\vspace{-.5cm}
\end {center}
Importantly, the nature of the quantum advantages established in Propositions \ref{prop1} \& \ref{prop2} are distinct from the advantage known in the communication complexity scenario \cite{Buhrman2010}. In communication complexity, the objective is to compute the value of a specific function whose inputs are distributed between two remote parties. The parties aim to achieve the goal using the least possible amount of communication between them. Canonical instances of such problems where quantum systems exhibit advantages over the classical counterparts are the task of quantum random access codes \cite{Wiesner1983,Ambainis2002,Ambainis2019}. Importantly, the quantum advantage in those problems relies on the exploitation of non-classical features of quantum systems during the encoding step as well as the decoding step. The sender prepares the quantum system in superposed states based on the inputs she receives, and the receiver, based on his inputs, selects a decoding measurement from a set of incompatible measurements. However, the scope for utilizing non-classical effects at the decoding step does not arise in the Frenkel-Weiner setup, as no input is provided to the receiver in this case. Interestingly, when considering a MAC, a new opportunity emerges at the decoding step where quantum effects can play a non-trivial role. The receiver can perform a global measurement, such as an entangled basis measurement, on the quantum systems received from different senders. This precisely happens while establishing the quantum advantage outlined in Propositions \ref{prop1} \& \ref{prop2}.

Naturally, the question arises: is an entanglement basis measurement necessary to achieve such an advantage? Interestingly, we will now demonstrate that this is not the case in general. To illustrate this, we consider a MAC involving three senders and one receiver. Senders are provided independent two-bit strings as inputs, {\it i.e.}, $X=Y=Z\equiv\{0,1\}^{\times2}$, while the receiver produces three-bit string outputs, {\it i.e.}, $A\equiv\{0,1\}^{\times3}$. Here also we take the reverse engineering approach to introduce the quantum strategy in $\mathbf{Q}_{ds}^{3}$ that leads us to the desired MAC (see Fig. \ref{fig2}). Since the receiver employs a decoding measurement in a product basis known as the SHIFT ensemble, we will refer to the resulting MAC as $\mathcal{N}^{3}_{shift}$, which can be compactly expressed as
\begin{align}
\mathcal{N}^{3}_{shift}&\equiv\left\{p({\bf a}|{\bf x},{\bf y},{\bf z})=\zeta_+^{\eta}\zeta_-^{(3-\eta)}\right\},\\
{\bf x},{\bf y},{\bf z}&\in \{0,1\}^{\times 2},~{\bf a}\in \{0,1\}^{\times 3},~\zeta_\pm:=\frac{1}{2}(1\pm\frac{1}{\sqrt{2}}),\nonumber\\
\mbox{where},~&p({\bf a}|{\bf x},{\bf y},{\bf z}):= 2 \delta_{3,\eta}-3 \delta_{2,\eta}-2 \delta_{1,\eta}-2 \delta_{0,\eta},\nonumber\\
\eta:=&\left\{\!\begin{aligned}
\delta_{x_1,a_1}+\delta_{y_1,a_2}+\delta_{z_1,a_3},~\text{if}~a_1=a_2=a_3,\\
\delta_{x_1,0}+\delta_{y_1,1}+\delta_{z_2,a_3},~~ \text{if}~a_1=a_2\neq a_3,\\
\delta_{x_1,1}+\delta_{y_2,a_2}+\delta_{z_1,0},~~ \text{if}~a_3=a_1\neq a_2,\\
\delta_{x_2,a_1}+\delta_{y_1,0}+\delta_{z_1,1},~~\text{if}~a_2=a_3\neq a_1.
\end{aligned}\right\}\label{eta}
\end{align}
Here $\delta$ denotes the Kronecker Delta symbol. Notably, the SHIFT measurement exhibits the phenomenon of `quantum nonlocality without entanglement' (QNWE) \cite{Bennett1999} (see also \cite{Bhattacharya2020}), and implementation of this measurement necessitates a global interaction among the three qubits \cite{Niset2006}. Our next result shows that $\mathcal{N}^{3}_{shift}$ cannot be simulable with the corresponding  classical strategies. 
\begin{figure}[t!]
\centering
\includegraphics[scale=0.5]{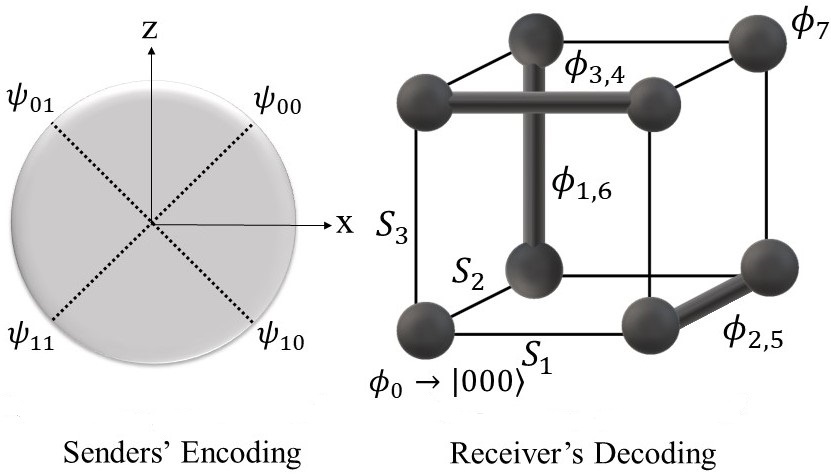}
\caption{$\mathbf{Q}_{ds}^{3}$ strategy simulating the MAC $\mathcal{N}^{3}_{shift}$. While the encoding states are symmetrically chosen from the $xz$-plane of the Bloch sphere, the receiver employs the decoding measurement in SHIFT basis \cite{Bennett1999}.}
\label{fig2}\vspace{-.5cm}
\end{figure}
\begin{proposition}\label{prop3}
The MAC $\mathcal{N}^{3}_{shift}$ cannot be simulated using any strategy from the set $\mathbf{C}^{3}_{cs}$.   
\end{proposition}
\begin{proof}
    {\it Structure of the MAC:} $\mathcal{N}^{3}_{shift}$ involves three senders and one receiver. Each sender is provided independent two-bit strings as inputs and the receiver yields a three-bit string as output. Denoting the input stings of different senders respectively as ${\bf x}=x_1x_2$, ${\bf y}=y_1y_2$, ${\bf z}=z_1z_2$, and denoting the output string as ${\bf a}=a_1a_2a_3$, the MAC $\mathcal{N}^{3}_{shift}$ is specified by the conditional probabilities of obtaining the outputs given the inputs, {\it i.e.}, $\mathcal{N}^{3}_{shift}\equiv\{p({\bf a}|{\bf x},{\bf y},{\bf z})~|~x_i,y_i,z_i,a_k\in\{0,1\},~\forall~i\in\{1,2\}~\&~k\in\{1,2,3\}\}$. Here, we take a reverse engineering approach to specify the conditional probabilities using the desired quantum strategy in $\mathbf{Q}_{ds}^{3}$.  which defines the MAC $\mathcal{N}^{3}_{shift}$.Denoting the input string as ${\bf s}=s_1s_2\in\{{\bf x},{\bf y},{\bf z}\}$, encoding states of the respective senders are given by
\begin{align}
\psi_{{\bf s}}\equiv\ket{\psi_{\bf s}}\bra{\psi_{\bf s}}=\frac{1}{2} \left[\mathbf{I} + \frac{(-1)^{s_1}}{\sqrt{2}}\sigma_Z +\frac{(-1)^{s_2}}{\sqrt{2}}\sigma_X\right].\label{EncStates}
\end{align}
Important to note that the encoding states are symmetrically chosen from the $xz$-plane of the Bloch sphere. For decoding the receiver performs SHIFT basis measurement on the three qubits received from the senders and decodes the outcome according to the following strategy: 
\begin{align}
\left\{\!\begin{aligned}
\ket{\phi_0}&:=\ket{000}\mapsto000,~~~\ket{\phi_1}:=\ket{01-}\mapsto001,\\
\ket{\phi_2}&:=\ket{1-0}\mapsto 010,~~\ket{\phi_3}:=\ket{+01}\mapsto 011,\\
\ket{\phi_4}&:=\ket{-01}\mapsto 100,~~ \ket{\phi_5}:=\ket{1+0}\mapsto 101,\\
\ket{\phi_6}&:=\ket{01+}\mapsto 110,~~~~\ket{\phi_7}:=\ket{111}\mapsto 111
\end{aligned}\right\}.	\label{DecProj}
\end{align}
The $8 \times 64$ Stochastic Matrix of the MAC $\mathcal{N}^{3}_{shift}$ has some interesting symmetry. To note this symmetry, consider the case where all the senders have input $`00'$, {\it i.e.}, ${\bf x}={\bf y}={\bf z}=00$. The probabilities of different outcomes at receiver's end read as:
\begin{align*}
\left\{\!\begin{aligned}
p(000)&=\zeta_+^3,~~p(011)=p(101)=p(110)=\zeta_+^2\zeta_-\\
p(001)&=p(010)=p(100)=\zeta_-^2\zeta_+,~~p(111)=\zeta_-^3,
\end{aligned}\right\},
\end{align*}
where $\zeta_{\pm}:=\frac{1}{2}(1\pm\frac{1}{\sqrt{2}})$. In other words, for the case ${\bf x}={\bf y}={\bf z}=00$, exactly one outcome occurs with probability $\zeta_+^3$, exactly one occurs with probability $\zeta_-^3$, exactly three occur with probability $\zeta_-^2\zeta_+$, and the remaining three occur with probability $\zeta_-\zeta_+^2$. Interestingly, for any of the input-triples $({\bf x},{\bf y},{\bf z})\in\{0,1\}^{\times2}\times\{0,1\}^{\times2}\times\{0,1\}^{\times2}$, exactly the same pattern with the same probabilities hold (we call it the ``$1$-$3$-$3$-$1$" pattern). In other words, each of the columns of this $8 \times 64$ Stochastic Matrix $\mathcal{N}^{3}_{shift}$ are identical up-to some permutation of the rows.\\\\
{\it Quantum advantage:} To show that the MAC $\mathcal{N}^{3}_{shift}$ cannot be simulated by any strategy in $\mathbf{C}^{3}_{cs}$, we now construct a witness ${\bf W}_{shift}\equiv \{W_{{\bf a}|{\bf x},{\bf y},{\bf z}}~|~{\bf x},{\bf y},{\bf z}\in \{0,1\}^{\times 2},~{\bf a}\in \{0,1\}^{\times 3}\}$, with the entries given by 
\begin{align}
W_{{\bf a}|{\bf x},{\bf y},{\bf z}}&:= 2 \delta_{3,\eta}-3 \delta_{2,\eta}-2 \delta_{1,\eta}-2 \delta_{0,\eta}~~; \label{WitnessShift}
\end{align} 
where $\eta$ is defined in Eq.(5) of the main manuscript. The MAC $\mathcal{N}^{3}_{shift}$ yields the following payoff: 
\begin{align}
&{\bf W}_{shift}\left[\mathcal{N}^{3}_{shift}\right]\equiv{\bf W}_{shift}\cdot\mathcal{N}^{3}_{shift}:= \sum_{{\bf x},{\bf y},{\bf z},{\bf a}} W_{{\bf a}|{\bf x},{\bf y},{\bf z}} \times p({\bf a}|{\bf x},{\bf y},{\bf z})\nonumber\\
&=\sum_{{\bf x},{\bf y},{\bf z}} \left\{\sum_{{\bf a}}W_{{\bf a}|{\bf x},{\bf y},{\bf z}} \times p({\bf a}|{\bf x},{\bf y},{\bf z})\right\}\nonumber\\
&=\sum_{{\bf x},{\bf y},{\bf z}} \left\{\sum_{\eta=0}^3\left[1\times(+2) \delta_{3,\eta}+3\times(-3)\delta_{2,\eta}+3\times(-2)\delta_{1,\eta}+1\times(-2)\delta_{0,\eta} \right]\times \zeta_+^{\eta}\zeta_-^{(3-\eta)} \right\}\nonumber\\
&=\sum_{{\bf x},{\bf y},{\bf z}}\left\{2\zeta_+^3
-9\zeta_+^2\zeta_--6\zeta_-^2\zeta_+-2\zeta_-^3\right\}=64\times\left\{2\zeta_+^3
-9\zeta_+^2\zeta_--6\zeta_-^2\zeta_+-2\zeta_-^3\right\} \nonumber\\
&=10(5\sqrt{2}-6)\approx10.7107~.
\end{align}
Above calculation utilizes the feature of ``$1$-$3$-$3$-$1$" pattern. To establish the quantum advantage we our now left to show that any strategy in $\mathbf{C}^{3}_{cs}$ yield a lower than quantum value for the witness operator ${\bf W}_{shift}$. 

{\it Bound on the classical payoff:} Here prove that for any strategy in $\mathbf{C}^{3}_{cs}$ the payoff is upper bounded by $8$. Since $\mathbf{C}^{3}_{cs}$ forms a convex set and since the witness ${\bf W}_{shift}$ is linear, the optimal payoff will be achieved for some strategy belonging to $\mathbf{C}^{3}_{ds}$. A strategy in $\mathbf{C}^{3}_{ds}$ will lead to a $8\times64$ Stochastic Matrix $\mathcal{S}_d\equiv\{p_d({\bf a}|{\bf x},{\bf y},{\bf z})\}$, where all the entries $p_d({\bf a}|{\bf x},{\bf y},{\bf z})$ are `$0$' or `$1$'. For such a strategy the payoff reads as
\begin{align}
{\bf W}_{shift}\left[\mathcal{S}_d\right]= \sum_{{\bf x},{\bf y},{\bf z},{\bf a}} W_{{\bf a}|{\bf x},{\bf y},{\bf z}} \times p_d({\bf a}|{\bf x},{\bf y},{\bf z}) =\sum_{{\bf x},{\bf y},{\bf z},{\bf a}\atop p_d({\bf a}|{\bf x},{\bf y},{\bf z})=1} W_{{\bf a}|{\bf x},{\bf y},{\bf z}}~~.
\end{align}
Instead of considering all the deterministic strategies, we will try to find optimal deterministic decoding strategy for a given deterministic encoding. For a clear exposition we explicitly analyze one such case.   Consider the encoding strategy where all senders send the first bit of their respective strings to the receiver, {\it i.e.},
\begin{align}
\mathrm{E}_i=\mathrm{e}^1&:=\mathrm{E}_{bit}^{1^{st}}: \{00,01\}\mapsto0,~~~\{10,11\}\mapsto1;~~i\in\{1,2,3\}.\label{1stbit}
\end{align}
The procedure of finding the optimal decoding strategy for this encoding is described in Table \ref{tabshift}. It turns out that optimal decoding yields the payoff $8$.  

Up-to the freedom of relabeling, each sender can chose their deterministic encoding from the following set of eight deterministic encodings:
\fontsize{10}{5}
\begin{align}
\mathcal{E}\equiv\left\{\!\begin{aligned}
\mathrm{e}^0&:=\mathrm{E}_{const}: \{00,01,10,11\}\mapsto0,~\{\}\mapsto1;~~~
\mathrm{e}^1:=\mathrm{E}_{bit}^{1^{st}}: \{00,01\}\mapsto0,~~~~~\{10,11\}\mapsto1;\\
\mathrm{e}^2&:=\mathrm{E}_{bit}^{2^{nd}}: \{00,10\}\mapsto0,~~~~~\{01,11\}\mapsto1;~~~
\mathrm{e}^3:=\mathrm{E}_{par}: \{00,11\}\mapsto0,~~~~~\{01,10\}\mapsto1;\\
\mathrm{e}^4&:=\mathrm{E}^{1|3}_{00}: \{00\}\mapsto0,~~~~~~\{01,10,11\}\mapsto1;~~
\mathrm{e}^5:=\mathrm{E}^{1|3}_{01}: \{01\}\mapsto0,~~~~~~\{00,10,11\}\mapsto1;\\
\mathrm{e}^6&:=\mathrm{E}^{1|3}_{10}: \{10\}\mapsto0,~~~~~~\{00,01,11\}\mapsto1;~~
\mathrm{e}^7:=\mathrm{E}^{1|3}_{11}: \{11\}\mapsto0,~~~~~~\{00,01,10\}\mapsto1;
\end{aligned}\right\}.
\end{align}
\normalsize
\begin{center}
\begin{table}[t!]
\fontsize{8}{8}
\begin{tabular}
{|c|c|c|c|c|c|c|c|c|c|c|c|c|c|} 
\hline
\multirow{2}{*}{~$\alpha\beta\gamma$~} & \multirow{2}{*}{\footnotesize Possible Inputs} & \multirow{2}{*}{$\chi_{000}$} & \multirow{2}{*}{$\chi_{001}$} & \multirow{2}{*}{$\chi_{010}$} & \multirow{2}{*}{$\chi_{011}$} & \multirow{2}{*}{$\chi_{100}$} & \multirow{2}{*}{$\chi_{101}$} & \multirow{2}{*}{$\chi_{110}$} & \multirow{2}{*}{$\chi_{111}$} & \multirow{2}{*}{\footnotesize Optimal} & \multirow{2}{*}{\footnotesize Best } \\ 
\multirow{2}{*}{} & \multirow{2}{*}{$({\bf x},{\bf y},{\bf z})\equiv(x_1x_2,y_1y_2,z_1z_2)$} & \multirow{2}{*}{} & \multirow{2}{*}{} & \multirow{2}{*}{} & \multirow{2}{*}{} & \multirow{2}{*}{} & \multirow{2}{*}{} & \multirow{2}{*}{} & \multirow{2}{*}{} & \multirow{2}{*}{\footnotesize~~Decoding~~} & \multirow{2}{*}{\footnotesize Payoff} \\
\multirow{2}{*}{} & \multirow{2}{*}{} & \multirow{2}{*}{} & \multirow{2}{*}{} & \multirow{2}{*}{} & \multirow{2}{*}{} & \multirow{2}{*}{} & \multirow{2}{*}{} & \multirow{2}{*}{} & \multirow{2}{*}{} & \multirow{2}{*}{} & \multirow{2}{*}{} \\
\hline
\multirow{2}{*}{$000$} & ~$(00,00,00),~(00,00,01),~(00,01,00),~(00,01,01),$~ & \multirow{2}{*}{$16$} & \multirow{2}{*}{$-20$} & \multirow{2}{*}{$-20$} & \multirow{2}{*}{$-20$} & \multirow{2}{*}{$-20$} & \multirow{2}{*}{$-20$} & \multirow{2}{*}{$-20$} & \multirow{2}{*}{$-16$} & \multirow{2}{*}{$000$} & \multirow{2}{*}{$16$}\\ 
\multirow{2}{*}{} & $(01,00,00),(01,00,01),~(01,01,00),~(01,01,01)~$  & \multirow{2}{*}{} & \multirow{2}{*}{} & \multirow{2}{*}{} & \multirow{2}{*}{} & \multirow{2}{*}{} & \multirow{2}{*}{} & \multirow{2}{*}{} &\multirow{2}{*}{} & \multirow{2}{*}{} & \multirow{2}{*}{}\\
\hline
\multirow{2}{*}{$001$} & $(00,00,10),~(00,00,11),~(00,01,10),~(00,01,11),$ & \multirow{2}{*}{$-24$} & \multirow{2}{*}{$-20$} & \multirow{2}{*}{$-16$} & \multirow{2}{*}{$-4$} & \multirow{2}{*}{$-4$} & \multirow{2}{*}{$-16$} & \multirow{2}{*}{$-20$} & \multirow{2}{*}{$-16$} & \multirow{2}{*}{$011/100$} & \multirow{2}{*}{$-4$}\\ 
\multirow{2}{*}{} & $(01,00,10),~(01,00,11),~(01,01,10),~(01,01,11)~$  & \multirow{2}{*}{} & \multirow{2}{*}{} & \multirow{2}{*}{} & \multirow{2}{*}{} & \multirow{2}{*}{} & \multirow{2}{*}{} & \multirow{2}{*}{} & \multirow{2}{*}{} &\multirow{2}{*}{} & \multirow{2}{*}{}\\
\hline
\multirow{2}{*}{$010$} & $(00,10,00),~(00,10,01),~(00,11,00),~(00,11,01),$ & \multirow{2}{*}{$-24$} & \multirow{2}{*}{$-4$} & \multirow{2}{*}{$-20$} & \multirow{2}{*}{$-16$} & \multirow{2}{*}{$-16$} & \multirow{2}{*}{$-20$} & \multirow{2}{*}{$-4$} & \multirow{2}{*}{$-16$} & \multirow{2}{*}{$001/110$} & \multirow{2}{*}{$-4$}\\ 
\multirow{2}{*}{} & $(01,10,00),~(01,10,01),~(01,11,00),~(01,11,01)~$  & \multirow{2}{*}{} & \multirow{2}{*}{} & \multirow{2}{*}{} & \multirow{2}{*}{} & \multirow{2}{*}{} & \multirow{2}{*}{} & \multirow{2}{*}{} & \multirow{2}{*}{} & \multirow{2}{*}{}& \multirow{2}{*}{}\\
\hline
\multirow{2}{*}{$011$} & $(00,10,10),~(00,10,11),~(00,11,10),~(00,11,11),$ & \multirow{2}{*}{$-16$} & \multirow{2}{*}{$-4$} & \multirow{2}{*}{$-16$} & \multirow{2}{*}{$-20$} & \multirow{2}{*}{$-20$} & \multirow{2}{*}{$-16$} & \multirow{2}{*}{$-4$} & \multirow{2}{*}{$-24$} & \multirow{2}{*}{$001/110$} & \multirow{2}{*}{$-4$}\\ 
\multirow{2}{*}{} & $(01,10,10),~(01,10,11),~(01,11,10),~(01,11,11)~$  & \multirow{2}{*}{} & \multirow{2}{*}{} & \multirow{2}{*}{} & \multirow{2}{*}{} & \multirow{2}{*}{} & \multirow{2}{*}{} & \multirow{2}{*}{} & \multirow{2}{*}{} & \multirow{2}{*}{}& \multirow{2}{*}{}\\
\hline
\multirow{2}{*}{$100$} & $(10,00,00),~(10,00,01),~(10,01,00),~(10,01,01),$ & \multirow{2}{*}{$-24$} & \multirow{2}{*}{$-16$} & \multirow{2}{*}{$-4$} & \multirow{2}{*}{$-20$} & \multirow{2}{*}{$-20$} & \multirow{2}{*}{$-4$} & \multirow{2}{*}{$-16$} & \multirow{2}{*}{$-16$} & \multirow{2}{*}{$010/101$}& \multirow{2}{*}{$-4$}\\ 
\multirow{2}{*}{} & $(11,00,00),~(11,00,01),~(11,01,00),~(11,01,01)~$  & \multirow{2}{*}{} & \multirow{2}{*}{} & \multirow{2}{*}{} & \multirow{2}{*}{} & \multirow{2}{*}{} & \multirow{2}{*}{} & \multirow{2}{*}{} & \multirow{2}{*}{} & \multirow{2}{*}{}& \multirow{2}{*}{}\\
\hline
\multirow{2}{*}{$101$} & $(10,00,10),~(10,00,11),~(10,01,10),~(10,01,11),$ & \multirow{2}{*}{$-16$} & \multirow{2}{*}{$-16$} & \multirow{2}{*}{$-20$} & \multirow{2}{*}{$-4$} & \multirow{2}{*}{$-4$} & \multirow{2}{*}{$-20$} & \multirow{2}{*}{$-16$} & \multirow{2}{*}{$-24$} & \multirow{2}{*}{$011/100$}& \multirow{2}{*}{$-4$}\\ 
\multirow{2}{*}{} & $(11,00,10),~(11,00,11),~(11,01,10),~(11,01,11)~$  & \multirow{2}{*}{} & \multirow{2}{*}{} & \multirow{2}{*}{} & \multirow{2}{*}{} & \multirow{2}{*}{} & \multirow{2}{*}{} & \multirow{2}{*}{} & \multirow{2}{*}{} & \multirow{2}{*}{}& \multirow{2}{*}{}\\
\hline
\multirow{2}{*}{$110$} & $(10,10,00),~(10,10,01),~(10,11,00),~(10,11,01),$ & \multirow{2}{*}{$-16$} & \multirow{2}{*}{$-20$} & \multirow{2}{*}{$-4$} & \multirow{2}{*}{$-16$} & \multirow{2}{*}{$-16$} & \multirow{2}{*}{$-4$} & \multirow{2}{*}{$-20$} & \multirow{2}{*}{$-24$} & \multirow{2}{*}{$010/101$}& \multirow{2}{*}{$-4$}\\ 
\multirow{2}{*}{} & $(11,10,00),~(11,10,01),~(11,11,00),~(11,11,01)~$  & \multirow{2}{*}{} & \multirow{2}{*}{} & \multirow{2}{*}{} & \multirow{2}{*}{} & \multirow{2}{*}{} & \multirow{2}{*}{} & \multirow{2}{*}{} & \multirow{2}{*}{} & \multirow{2}{*}{}& \multirow{2}{*}{}\\
\hline
\multirow{2}{*}{$111$} & $(10,10,10),~(10,10,11),~(10,11,10),~(10,11,11),$ & \multirow{2}{*}{$-16$} & \multirow{2}{*}{$-20$} & \multirow{2}{*}{$-20$} & \multirow{2}{*}{$-20$} & \multirow{2}{*}{$-20$} & \multirow{2}{*}{$-20$} & \multirow{2}{*}{$-20$} & \multirow{2}{*}{$16$} & \multirow{2}{*}{$111$}& \multirow{2}{*}{$16$}\\ 
\multirow{2}{*}{} & $(11,10,10),~(11,10,11),~(11,11,10),~(11,11,11)~$  & \multirow{2}{*}{} & \multirow{2}{*}{} & \multirow{2}{*}{} & \multirow{2}{*}{} & \multirow{2}{*}{} & \multirow{2}{*}{} & \multirow{2}{*}{} & \multirow{2}{*}{} & \multirow{2}{*}{}& \multirow{2}{*}{}\\
\hline
\multicolumn{11}{|c|}{\bf Total Payoff (TP)} & \textbf{8}\\
\hline
\end {tabular}
\caption{In the first column, $\alpha\beta\gamma$ denotes the message receiver gets from the senders. Second column lists the set of inputs $({\bf x},{\bf y}, {\bf z})$ leading to the corresponding communication $\alpha\beta\gamma$ due to the encoding (\ref{1stbit}). The symbols $\chi_{ijk}$ in columns (3-10) denote the payoff obtained for the communication $\alpha\beta\gamma$, given that the receiver decodes $\alpha\beta\gamma\mapsto ijk$. For instance, if the receiver decodes the communication `$000$' to output `$000$', {\it i.e.}, $(\alpha\beta\gamma=000)\mapsto (ijk=000)$, then according to the witness ${\bf W}_{shift}$ as defined in Eq.(\ref{WitnessShift}), the payoff for all the possible inputs $\{(00,00,00),(00,00,01),(00,01,00),(00,01,01),(01,00,00),(01,00,01),(01,01,00),(01,01,01)\}$ is $+2$, leading to $\chi_{000}=8\times2=16$. On the other hand, if the receiver decodes $(\alpha\beta\gamma=000)\mapsto $ $(ijk=001)$, then the inputs $\{(00,00,01),(00,01,01),(01,00,01),(01,01,01)\}$ have payoff $-3$ each, and the remaining four possible inputs have payoff $-2$ each. Thus, $\chi_{001}=4\times(-3)+4\times(-2)=-20$. Rightmost column lists the best possible payoff for the communication $\alpha\beta\gamma$.}\label{tabshift}
\end {table}
\end {center}
\normalsize
While in Table \ref{tabshift} we analyze the case where three senders follow the encoding $\left(\mathrm{e}^1,\mathrm{e}^1,\mathrm{e}^1\right)$, similar analysis can be accomplished efficiently for $8^3$ different encoding triples $\left(\mathrm{e}^p,\mathrm{e}^q,\mathrm{e}^r\right)\in\mathcal{E}\times\mathcal{E}\times\mathcal{E}$, where $p,q,r\in\{0,1,\cdots,7\}$. In Appendix \ref{applast}, we list the optimal Total payoff (TP) for all these encodings, which in turn proves the claim. 
\end{proof}
\begin{remark}
    Proposition \ref{prop3} is remarkable from another perspective as well. As highlighted by Bennett and Shor \cite{Bennett1998}, for a quantum channel, four basic types of classical capacities can be defined, that correspond to utilization of either product or entangled states at the input, and product or entangled measurements at the output. Although these capacities become identical for a perfect quantum channel, in presence of noise, leveraging entanglement in encoding and decoding can provide a higher capacity compared to using only product encoding-decoding \cite{Fuchs1997,Holevo1998,Hastings2009}. However, for product decoding, such an advantage is absent even when entangled states are employed in the encoding step \cite{King2001}. Nevertheless, Proposition \ref{prop3} reveals that even a product basis measurement, which exhibits QNWE can prove beneficial for simulating a MAC.
\end{remark}

In the next, we describe simulation of the aforesaid MACs with 1 c-bit channel from each of the senders to the receiver, provided the channels get assisted with entanglement. 
\begin{proposition}\label{prop4}
The MACs $\mathcal{N}^{2}_{PBR}$, $\mathcal{N}^{3}_{shift}$, and $\mathcal{N}^{2}_{m}$ can all be simulated using $1$-cbit communication from each sender to the receiver, provided each communication line is assisted by a two-qubit maximally entangled state.  
\end{proposition}
The claim simply follows from the familiar `remote state preparation' protocol \cite{Lo2000,Pati2000,Bennett2001}, since in all the  cases the encoding states are chosen from great circles of the Bloch sphere. The decoding step goes as that of the qubit based protocols. Proposition \ref{prop4} establishes nontrivial usage of quantum entanglement in network communication scenario.

\section{Discussions} Quantum advantages are elusive, difficult to establish, and often constrained by fundamental no-go theorems. For instance, quantum computing can provide speedup over the classical computers for a range of problems. However, the set of functions computable using quantum mechanics is precisely equivalent to what can be computed using classical physics. Similarly, in point-to-point communication Holevo-Frenkel-Weiner results limit the utility of quantum systems when no entanglement is shared between sender and receiver. In this work we show that such limitations do no hold true in network communication setup. Particularly, in simulating multiple-sender-to-one-receiver channels, our study uncovers a novel advantage of quantum communication over the the corresponding classical resource. Notably, the presented advantage is distinct from other recent studies on MAC \cite{Quek2017,Leditzky2020,Yun2020}, where it has been shown that nonlocal correlations shared among the distant senders can lead to higher channel capacities. As previously highlighted, our construction in Proposition \ref{prop1} draws inspiration from the renowned PBR theorem. This theorem establishes quantum wave-functions to be {\it $\psi$-ontic}, implying a direct correspondence with reality \cite{Harrigan2010}. It would be, therefore, intriguing to explore more direct connection between the $\psi$-onticity of quantum wave-functions and the quantum advantage reported here. It is noteworthy that the quantum advantages in Propositions \ref{prop1} \& \ref{prop2} do not rely on incompatible measurements at the decoding step; rather, measurements involving entangled projectors are utilized. Such measurements are known to play a pivotal role in the phenomenon of `quantum nonlocality without inputs' \cite{Renou2019}. The deeper connection between the reported quantum advantage and the concept of network nonlocality \cite{Tavakoli2022} is worth exploring. On the other hand, our construction also provides a way for semi-device-independent certification of entangled measurements \cite{Vertesi2011,Bennet2014,Supic2023}.  

The MAC presented in Proposition \ref{prop3} underscores the intricate role of QNWE in establishing the qubit advantage over the c-bit in network communication scenarios. Numerous other product bases are documented in literature that exhibit the QNWE phenomenon, with recent studies even introducing various variants of this phenomenon \cite{Halder2019, Rout2019, Ghosh2022}. Exploring these constructions to establish quantum advantages in the MAC scenario would be highly intriguing. Lastly, Proposition \ref{prop4} demonstrates that all the reported advantages of the qubit channel over the c-bit vanish when the latter is assisted with entanglement. It would be quite interesting to formulate a scenario where the qubit channel maintains an advantage over an entanglement-assisted c-bit channel.

\appendix
\section{Proof of Proposition \ref{prop1}}\label{appa}
Notably, few of the conditional probabilities in $\mathcal{N}^{2}_{PBR}$ are zero. In particular,
\begin{align}
p(\mathbf{a}|\mathbf{x},\mathbf{y})=0,~~\mbox{when}~a_1=\overline{x_1\oplus y_1}~\wedge~a_2=\overline{x_2\oplus y_2}.\label{zero}
\end{align}
Represented as a stochastic matrix, the MAC $\mathcal{N}^{2}_{PBR}$ reads as
\fontsize{7}{8}
\begin{align*}
\label{N2}
\mathcal{N}^{_{2}}_{PBR}
\equiv
\begingroup
\setlength{\tabcolsep}{1pt} 
\renewcommand{\arraystretch}{1} 
\begin{array}{c||c|c|c|c|c|c|c|c|c|c|c|c|c|c|c|c|} 
\mathbf{a}\textbackslash\mathbf{x},\mathbf{y}& 00,00 & 00,01 & 00,10  & 00,11& 01,00 & 01,01  & 01,10  & 01,11& 10,00 & 10,01  & 10,10 & 10,11& 11,00 & 11,01  & 11,10 & 11,11 \\ \hline\hline
00& 1/2 & 1/4 & 1/4 & \color{blue}{0}   & 1/4 & 1/2 & \color{blue}{0}   & 1/4 & 1/4 & \color{blue}{0}   & 1/2 & 1/4 & \color{blue}{0}   & 1/4 & 1/4 & 1/2 \\ \hline
01& 1/4 & 1/2 & \color{blue}{0}   & 1/4 & 1/2 & 1/4 & 1/4 & \color{blue}{0}   & \color{blue}{0}   & 1/4 & 1/4 & 1/2 & 1/4 & \color{blue}{0} & 1/2 & 1/4 \\ \hline
10& 1/4 & \color{blue}{0}   & 1/2 & 1/4 & \color{blue}{0}   & 1/4 & 1/4 & 1/2 & 1/2 & 1/4 & 1/4 & \color{blue}{0} 
  & 1/4 & 1/2 & \color{blue}{0}   & 1/4 \\ \hline
11& \color{blue}{0}   & 1/4 & 1/4 & 1/2 & 1/4 & \color{blue}{0}   & 1/2 & 1/4 & 1/4 & 1/2 & \color{blue}{0}   & 1/4 & 1/2 & 1/4 & 1/4 & \color{blue}{0} \\ \hline
\end{array}~.
\endgroup
\end{align*}
\normalsize
A classical deterministic strategy $\mathcal{S}\equiv(\mathrm{E}_1,\mathrm{E}_2,\mathrm{D})\in\mathbf{C}_{ds}^{2}$ aiming to simulate $\mathcal{N}^{2}_{PBR}$ must yield zero value for the conditional probabilities in Eq.(\ref{zero}). The encoding strategies for the senders are  deterministic functions from their two-bit input strings to one-bit message, {\it i.e.}
\begin{subequations}
\begin{align}
\mathrm{E}_1&:~\{0,1\}^{\otimes2}\ni\mathbf{x}\equiv x_1x_2\mapsto\alpha\in\{0,1\},\\
\mathrm{E}_2&:~\{0,1\}^{\otimes2}\ni\mathbf{y}\equiv y_1y_2\mapsto\beta\in\{0,1\}.
\end{align}
\end{subequations}
Accordingly, the senders send their respective one-bit messages $\alpha$ and $\beta$ to the receiver using $1$-bit classical channel. A deterministic decoding strategy at the receiver's end is a function from the communications received to the output, {\it i.e.}
\begin{align}
\mathrm{D}:~\alpha\times\beta\mapsto \mathbf{a}\equiv a_1a_2\in\{0,1\}^{\times2}. 
\end{align}
The following observations are crucial.
\begin{center}
\begin{table}[b!]
\begin{tabular}{|c|c|c|c|c|c|}
\hline
&&&&&\\
~~~~$\alpha\beta$~~~~& ~~$x_1x_2$~~ & ~~$y_1y_2$~~ & ~~$x_1\oplus y_1$~~ & ~~$x_2\oplus y_2$~~ & ~~$\overline{x_1\oplus y_1}~\overline{x_2\oplus y_2}$~~\\ 
&&&&&\\\hline\hline
\multirow{6}{*}{$00$} & $00$ & $00$ & $0$ & $0$ & $11$ \\ \cline{2-6} 
& $00$ & $01$ & $0$ & $1$ & $10$ \\ 
\cline{2-6}
& $00$ & $10$ & $1$ & $0$ & $01$ \\ 
\cline{2-6} 
& $01$ & $00$ & $0$ & $1$ & $10$ \\ 
\cline{2-6}
& $01$ & $01$ & $0$ & $0$ & $11$ \\ 
\cline{2-6}
& $01$ & $10$ & $1$ & $1$ & $00$  \\ \hline
\end{tabular}
\captionsetup{justification=centering}
\caption{An instance of disallowed encoding.}\label{disa}
\end{table}
\end{center}

\begin{observation}\label{obs1}
For the fixed encoding strategies followed by the senders, the conditions of Eq.(\ref{zero}) impose restriction on the allowed decodings by the receiver. For instance, consider the case where the senders use the encoding
\begin{subequations}
\begin{align}
\mathrm{E}_1&:~00\mapsto 0~\&~\{01,10,11\}\mapsto1;\\
\mathrm{E}_2&:~00\mapsto 0~\&~\{01,10,11\}\mapsto1.
\end{align}
\end{subequations}
The receiver cannot employ the decoding strategies $\mathrm{D}:00\mapsto11$, as it yields $p(11|00,00)=1$, a violation of the requirement (\ref{zero}). Therefore, for the given encoding this particular decoding is not allowed. 
\begin{center}
\begin{table}[t!]
\begin{tabular}{|c|c|c|c|c|c|c|}
\hline
&&&&&&Compatible\\
~~~~$\alpha\beta$~~~~& ~~$x_1x_2$~~ & ~~$y_1y_2$~~ & ~~$x_1\oplus y_1$~~ & ~~$x_2\oplus y_2$~~ & ~~$\overline{x_1\oplus y_1}~\overline{x_2\oplus y_2}$~~& Decoding\\&&&&&&$a_1a_2$\\ \hline\hline
\multirow{4}{*}{$00$} & $00$ & $00$ & $0$ & $0$ & $11$ & \multirow{4}{*}{$00~\mbox{or}~01$} \\ \cline{2-6} 
& $00$ & $01$ & $0$ & $1$ & $10$ & \multirow{4}{*}{}\\ 
\cline{2-6}
& $01$ & $00$ & $0$ & $1$ & $10$ & \multirow{4}{*}{}\\ 
\cline{2-6} 
& $01$ & $01$ & $0$ & $0$ & $11$ & \multirow{4}{*}{}\\ 
\hline
\multirow{4}{*}{$01$} & $00$ & $10$ & $1$ & $0$ & $01$ & \multirow{4}{*}{$11~\mbox{or}~10$} \\ \cline{2-6} 
& $00$ & $11$ & $1$ & $1$ & $00$ & \multirow{4}{*}{} \\ 
\cline{2-6}
& $01$ & $10$ & $1$ & $1$ & $00$ & \multirow{4}{*}{}\\ 
\cline{2-6} 
& $01$ & $11$ & $1$ & $0$ & $01$ & \multirow{4}{*}{}\\ 
\hline
\multirow{4}{*}{$10$} & $10$ & $00$ & $1$ & $0$ & $01$ & \multirow{4}{*}{$11~\mbox{or}~10$} \\ \cline{2-6} 
& $10$ & $01$ & $1$ & $1$ & $00$ & \multirow{4}{*}{} \\ 
\cline{2-6}
& $11$ & $00$ & $1$ & $1$ & $00$ & \multirow{4}{*}{}\\ 
\cline{2-6} 
& $11$ & $01$ & $1$ & $0$ & $01$ & \multirow{4}{*}{}\\ 
\hline
\multirow{4}{*}{$11$} & $10$ & $10$ & $0$ & $0$ & $11$ & \multirow{4}{*}{$00~\mbox{or}~01$} \\ \cline{2-6} 
& $10$ & $11$ & $0$ & $1$ & $10$ & \multirow{4}{*}{}\\ 
\cline{2-6}
& $11$ & $10$ & $0$ & $1$ & $10$ & \multirow{4}{*}{}\\ 
\cline{2-6} 
& $11$ & $11$ & $0$ & $0$ & $11$ & \multirow{4}{*}{}\\ 
\hline
\end{tabular}
\captionsetup{justification=centering}
\caption{Allowed decoding strategies for the encoding $\mathrm{E}_1=\mathrm{e}^1=\mathrm{E}_2$.}\label{ad}
\end{table}
\end{center}
\end{observation}
\begin{observation}\label{obs2}
For certain encoding strategies, there does not exist any decoding strategy compatible with the requirement (\ref{zero}). To see an explicit example, consider the following encodings
\begin{subequations}
\begin{align}
\mathrm{E}_1&:~\{00,01\}\mapsto 0~\&~\{10,11\}\mapsto 1;\\
\mathrm{E}_2&:~\{00,01,10\}\mapsto 0~\&~11\mapsto1.
\end{align}
\end{subequations}
In Table \ref{disa}, we analyze the case when the receiver gets $\alpha=0$ and $\beta=0$ from the two senders, respectively. While decoding, the message $\alpha\beta=00$ must be decoded as one of four possible outcomes $\mathbf{a}=a_1a_2\in\{0,1\}^{\times2}$. As evident from the last column of Table \ref{disa}, whatever output is decoded for the message $\alpha\beta=00$ the condition (\ref{zero}) gets violated. 
\end{observation}
These observations subsequently lead us to the following general Lemma.
\begin{Lemma}\label{lemma1}
Encoding strategics where the four inputs are grouped in two disjoint set with equal cardinalities for both the senders are the only possible encodings compatible for simulating the MAC $\mathcal{N}^{2}_{PBR}$. 
\end{Lemma}
\begin{proof}
A generic encoding, $\mathrm{E}_1:\mathbf{x}\mapsto\alpha$, employed by the sender $S_1$ (and similarly for the sender $S_2$) is a partition of the input strings $\mathbf{x}\in\{0,1\}^{\times2}$ into two disjoint sets $\mathbf{x}^\alpha$, {\it i.e.},
\begin{align}
\mathbf{x}=\mathbf{x}^0\cup\mathbf{x}^1,~~\mbox{such~that}~~\mathbf{x}^0\cap\mathbf{x}^1=\emptyset.
\end{align}
Let, cardinalities of the sets $\mathbf{x}^0$ and $\mathbf{x}^1$ be $\mathbf{c}$ and $4-\mathbf{c}$, respectively, with $\mathbf{c}\in \{0,1,2,3,4\}$. As the partitions with $|\mathbf{x}^0|=\mathbf{c}$ and $|\mathbf{x}^0|=4-\mathbf{c}$ are same under relabeling, it is thus sufficient to analyze the cases $\mathbf{c}\in \{0,1,2\}$.

Let us first consider the case where $|\mathbf{x}^0|=2$ and $|\mathbf{y}^0|=1$. The first sender can employ one of the following strategies, 
\begin{subequations}
\begin{align}
\mathrm{e}^1&:=\mathrm{E}_{bit}^{1^{st}}: \{00,01\}\mapsto0,~~~\{10,11\}\mapsto1; \\
\mathrm{e}^2&:=\mathrm{E}_{bit}^{2^{nd}}: \{00,10\}\mapsto0,~~~\{01,11\}\mapsto1; \\
\mathrm{e}^3&:=\mathrm{E}_{par}: \{00,11\}\mapsto0,~~~\{01,10\}\mapsto1;
\end{align}
\end{subequations}
and the encoding of the second sender can be represented as 
\begin{align}
\mathrm{E}^{1|3}_{i,j}: \mathrm{ij}\mapsto0,~~~\{\Bar{\mathrm{i}}\mathrm{j},\mathrm{i}\Bar{\mathrm{j}},\Bar{\mathrm{i}}\Bar{\mathrm{j}}\}\mapsto1,~~~\mbox{where}~~\mathrm{i},\mathrm{j}\in\{0,1\}.
\end{align}
As it turns out, whatever encodings are followed by the senders, a situation like Observation \ref{obs2} arises. Same will be the case whenever $|\mathbf{x}^0|=1$ \& $|\mathbf{y}^0|=2$ as well as $|\mathbf{x}^0|<2$ \& $|\mathbf{y}^0|<2$. This concludes the proof.  
\end{proof}
Among the allowed encodings $\mathrm{E}_1,\mathrm{E}_2\in\left\{\mathrm{e}^1,\mathrm{e}^2,\mathrm{e}^3\right\}$, following observations are further relevant. 
\begin{observation}\label{obs3}
There does not exist a valid decoding compatible with (\ref{zero}) whenever the senders employ encoding $\mathrm{E}_1=\mathrm{e}^u$ and $\mathrm{E}_2=\mathrm{e}^v$, with $u\neq v$ and $u,v,\in\{1,2,3\}$.  
\end{observation}
\begin{observation}\label{obs4}
For each of the encoding $\mathrm{E}_1=\mathrm{e}^u=\mathrm{E}_2$, there are $16$ possible decoding $\mathrm{d}^w_u$ leading to valid strategies $\mathcal{S}^{uw}\equiv\left(\mathrm{e}^u,\mathrm{e}^u,\mathrm{d}^w_u\right)$ that are compatible with (\ref{zero}); $u\in\{1,2,3\}$ and $w\in\{1,\cdots,16\}$. For instance, for the encoding $\mathrm{E}_1=\mathrm{e}^1=\mathrm{E}_2$ the allowed decoding strategies are listed in Table \ref{ad}.
\end{observation}
\begin{remark}\label{rem1}
The senders and receiver can further employ a strategy $\mathcal{S}^u[p_w]=\sum_w p_w\mathcal{S}^{uw}$, where $\{p_w\}_w$ denotes a probability distribution. Such a strategy can be employed using local randomness at receiver's end. As we note, none of these strategies can simulate the MAC $\mathcal{N}^{2}_{PBR}$. This can also be argued with simple reasoning. In this strategy, the encoding of the sender is fixed and there are only $16$ compatible decoding strategies exist. Any communication $\alpha\beta$ can be decoded into two possible outcomes which in turn leads to the impossibility of the simulation of MAC $\mathcal{N}^{2}_{PBR}$. For instance, it is clear from the above table that mixing all compatible decoding strategies for this fixed encoding will results in $p(10|00,00)=0$.
\end{remark}
\begin{remark}\label{rem2}
The senders and receiver can also employ a strategy $\mathcal{S}[p_{uw}]=\sum_{uw} p_{uw}\mathcal{S}^{uw}$, where $\{p_{uw}\}_{uw}$ denotes classical correlation shared among the senders and receiver. As it turns out the strategy $\mathcal{S}[\frac{1}{32}]=\sum_{u=1}^{2}\sum_{w=1}^{16}\frac{1}{32}\mathcal{S}^{uw}$ simulates the MAC $\mathcal{N}^{2}_{PBR}$. Importantly, this strategy cannot be implemented through $1$-cbit channel from each sender to the receiver assisted with SR resource $\$_{RS_1}\cup\$_{RS_1}$. However, it can be implemented if the resource $\$_{G}$ is allowed to be shared.  
\end{remark}

\section{Witness operators appeared in Proposition \ref{prop2}}\label{appb}
The linear inequalities for each $m\in\{5,\cdots,9\}$ is given by,
\begin{equation}
\mathbf{w}_m[\mathbf{p}_c]:=\mathbf{w}_m\cdot\mathbf{p}_c:= \sum_{x,y} w_m^{x,y}p_c(x,y) \le k_m,
\end{equation}
 The explicit forms of the witness operators $\mathbf{w}_m$ are provided below:
\begin{align*}
\mathbf{w}_5&=\begin{bmatrix}
  0 & 1 & -3 & -5 & 3 \\
  3 & 0 & 1 & -3 & -5 \\
  -5 & 3 & 0 & 1 & -3 \\
  -3 & -5 & 3 & 0 & 1 \\
  1 & -3 & -5 & 3 & 0 \\
\end{bmatrix};\hspace{1.9cm}
\mathbf{w}_6=\begin{bmatrix}
  0 & 4 & -4 & -4 & -5 & 2 \\
  2 & 0 & 4 & -4 & -4 & -2 \\
  -5 & 2 & 0 & 4 & -7 & -10 \\
  -4 & -5 & 2 & 0 & 4 & -4 \\
  -4 & -4 & -2 & 2 & 0 & 4 \\
  4 & -7 & -4 & -5 & 2 & 0 \\
\end{bmatrix};\\
\mathbf{w}_7&=\begin{bmatrix}
  0 & 1 & -1 & -1 & -1 & -1 & 1 \\
  1 & 0 & 1 & -1 & -1 & -1 & -1 \\
  -1 & 1 & 0 & 1 & -1 & -1 & -1 \\
  -1 & -1 & 1 & 0 & 1 & -1 & -1 \\
  -1 & -1 & -1 & 1 & 0 & 1 & -1 \\
  -1 & -1 & -1 & -1 & 1 & 0 & 1 \\
  1 & -1 & -1 & -1 & -1 & 1 & 0 \\
\end{bmatrix};~~~
\mathbf{w}_8=\begin{bmatrix}
  0 & 0 & 0 & 0 & 0 & 0 & 0 & 0 \\
  1 & 0 & 7 & 0 & -11 & 0 & -5 & 0 \\
  0 & 0 & 0 & 0 & 0 & 0 & 0 & 0 \\
  -5 & 0 & 1 & 0 & 7 & 0 & -11 & 0 \\
  0 & 0 & 0 & 0 & 0 & 0 & 0 & 0 \\
  -11 & 0 & -5 & 0 & 1 & 0 & 7 & 0 \\
  0 & 0 & 0 & 0 & 0 & 0 & 0 & 0 \\
  7 & 0 & -11 & 0 & -5 & 0 & 1 & 0 \\
\end{bmatrix};\\
&\hspace{2cm}\mathbf{w}_9=\begin{bmatrix}
  0 & 0 & 1 & -1 & -1 & -1 & -1 & 1 & 0 \\
  0 & 0 & 0 & 1 & -1 & -1 & -1 & -1 & 1 \\
  1 & 0 & 0 & 0 & 1 & -1 & -1 & -1 & -1 \\
  -1 & 1 & 0 & 0 & 0 & 1 & -1 & -1 & -1 \\
  -1 & -1 & 1 & 0 & 0 & 0 & 1 & -1 & -1 \\
  -1 & -1 & -1 & 1 & 0 & 0 & 0 & 1 & -1 \\
  -1 & -1 & -1 & -1 & 1 & 0 & 0 & 0 & 1 \\
  1 & -1 & -1 & -1 & -1 & 1 & 0 & 0 & 0 \\
  0 & 1 & -1 & -1 & -1 & -1 & 1 & 0 & 0 \\
\end{bmatrix}.~~
\end{align*}

\begin{table}[t!]
\section{Payoff Table: Proposition \ref{prop3}}\label{applast}
\vspace{-2cm}
\begin{scriptsize}
\begin{minipage}{.15\linewidth}
\vspace{-2cm}
\centering
\begin{tabular}{||c|c|c|c||}
\hline
\multicolumn{3}{|c|}{\bf Encoding} & \multirow{2}{*}{\bf TP} \\ 
\cline{1-3} 
$\mathrm{e}^i$ & $\mathrm{e}^j$ & $\mathrm{e}^k$ & \multirow{2}{*}{}\\
\hline\hline
\multirow{2}{*}{$0$}& $0$ & $0,3$ &\multirow{4}{*}{ \rotatebox{90}{$-120$} }\\\cline{2-3}
\multirow{ 2 }{*}{}& $3$ & $0,3$ &\multirow{ 4 }{*}{}\\\cline{1-3}
\multirow{ 2 }{*}{ $3$ }& $0$ & $0,3$ &\multirow{ 4 }{*}{}\\\cline{2-3}
\multirow{ 2 }{*}{}& $3$ & $0,3$ &\multirow{ 4 }{*}{}\\
\hline
\multirow{ 6 }{*}{ $0$ }& $0$ & $4, 5, 6, 7$ &\multirow{ 20 }{*}{ \rotatebox{90}{$-108$} }\\\cline{2-3}
\multirow{ 6 }{*}{}& $3$ & $4, 5, 6, 7$ &\multirow{ 20 }{*}{}\\\cline{2-3}
\multirow{ 6 }{*}{}& $4$ & $0, 3$ &\multirow{ 20 }{*}{}\\\cline{2-3}
\multirow{ 6 }{*}{}& $5$ & $0, 3$ &\multirow{ 20 }{*}{}\\\cline{2-3}
\multirow{ 6 }{*}{}& $6$ & $0, 3$ &\multirow{ 20 }{*}{}\\\cline{2-3}
\multirow{ 6 }{*}{}& $7$ & $0, 3$ &\multirow{ 20 }{*}{}\\\cline{1-3}
\multirow{ 6 }{*}{ $3$ }& $0$ & $4, 5, 6, 7$ &\multirow{ 20 }{*}{}\\\cline{2-3}
\multirow{ 6 }{*}{}& $3$ & $4, 5, 6, 7$ &\multirow{ 20 }{*}{}\\\cline{2-3}
\multirow{ 6 }{*}{}& $4$ & $0, 3$ &\multirow{ 20 }{*}{}\\\cline{2-3}
\multirow{ 6 }{*}{}& $5$ & $0, 3$ &\multirow{ 20 }{*}{}\\\cline{2-3}
\multirow{ 6 }{*}{}& $6$ & $0, 3$ &\multirow{ 20 }{*}{}\\\cline{2-3}
\multirow{ 6 }{*}{}& $7$ & $0, 3$ &\multirow{ 20 }{*}{}\\\cline{1-3}
\multirow{ 2 }{*}{ 4 }& $0$ & $0, 3$ &\multirow{ 20 }{*}{}\\\cline{2-3}
\multirow{ 2 }{*}{}& $3$ & $0, 3$ &\multirow{ 20 }{*}{}\\\cline{1-3}
\multirow{ 2 }{*}{ 5 }& $0$ & $0, 3$ &\multirow{ 20 }{*}{}\\\cline{2-3}
\multirow{ 2 }{*}{}& $3$ & $0, 3$ &\multirow{ 20 }{*}{}\\\cline{1-3}
\multirow{ 2 }{*}{ $6$ }& $0$ & $0, 3$ &\multirow{ 20 }{*}{}\\\cline{2-3}
\multirow{ 2 }{*}{}& $3$ & $0, 3$ &\multirow{ 20 }{*}{}\\\cline{1-3}
\multirow{ 2 }{*}{ $7$ }& $0$ & $0, 3$ &\multirow{ 20 }{*}{}\\\cline{2-3}
\multirow{ 2 }{*}{}& $3$ & $0, 3$ &\multirow{ 20 }{*}{}\\
\hline
\multirow{ 4 }{*}{ $0$ }& $0$ & $1, 2$ &\multirow{ 13 }{*}{ \rotatebox{90}{$-96$} }\\\cline{2-3}
\multirow{ 4 }{*}{}& $1$ & $0, 3$ &\multirow{ 13 }{*}{}\\\cline{2-3}
\multirow{ 4 }{*}{}& $2$ & $0, 2, 3$ &\multirow{ 13 }{*}{}\\\cline{2-3}
\multirow{ 4 }{*}{}& $3$ & $1, 2$ &\multirow{ 13 }{*}{}\\\cline{1-3}
\multirow{ 2 }{*}{ $1$ }& $0$ & $0, 3$ &\multirow{ 13 }{*}{}\\\cline{2-3}
\multirow{ 2 }{*}{}& $3$ & $0, 3$ &\multirow{ 13 }{*}{}\\\cline{1-3}
\multirow{ 3 }{*}{ $2$ }& $0$ & $0, 2, 3$ &\multirow{ 13 }{*}{}\\\cline{2-3}
\multirow{ 3 }{*}{}& $2$ & $0, 2, 3$ &\multirow{ 13 }{*}{}\\\cline{2-3}
\multirow{ 3 }{*}{}& $3$ & $0, 2, 3$ &\multirow{ 13 }{*}{}\\\cline{1-3}
\multirow{ 4 }{*}{ 3 }& $0$ & $1, 2$ &\multirow{ 13 }{*}{}\\\cline{2-3}
\multirow{ 4 }{*}{}& $1$ & $0, 3$ &\multirow{ 13 }{*}{}\\\cline{2-3}
\multirow{ 4 }{*}{}& $2$ & $0, 2, 3$ &\multirow{ 13 }{*}{}\\\cline{2-3}
\multirow{ 4 }{*}{}& $3$ & $1, 2$ &\multirow{ 13 }{*}{}\\
\hline
\multirow{ 3 }{*}{ $0$ }& $2$ & $4, 5$ &\multirow{ 16 }{*}{ \rotatebox{90}{$-92$} }\\\cline{2-3}
\multirow{ 3 }{*}{}& $6$ & $2$ &\multirow{ 16 }{*}{}\\\cline{2-3}
\multirow{ 3 }{*}{}& $7$ & $2$ &\multirow{ 16 }{*}{}\\\cline{1-3}
\multirow{ 4 }{*}{ $2$ }& $0$ & $6, 7$ &\multirow{ 16 }{*}{}\\\cline{2-3}
\multirow{ 4 }{*}{}& $3$ & $6, 7$ &\multirow{ 16 }{*}{}\\\cline{2-3}
\multirow{ 4 }{*}{}& $4$ & $0, 3$ &\multirow{ 16 }{*}{}\\\cline{2-3}
\multirow{ 4 }{*}{}& $5$ & $0, 3$ &\multirow{ 16 }{*}{}\\\cline{1-3}
\multirow{ 3 }{*}{ $3$ }& $2$ & $4, 5$ &\multirow{ 16 }{*}{}\\\cline{2-3}
\multirow{ 3 }{*}{}& $6$ & $2$ &\multirow{ 16 }{*}{}\\\cline{2-3}
\multirow{ 3 }{*}{}& $7$ & $2$ &\multirow{ 16 }{*}{}\\\cline{1-3}
\multirow{ 2 }{*}{ $4$ }& $0$ & $2$ &\multirow{ 16 }{*}{}\\\cline{2-3}
\multirow{ 2 }{*}{}& $3$ & $2$ &\multirow{ 16 }{*}{}\\\cline{1-3}
\multirow{ 2 }{*}{ $5$ }& $0$ & $2$ &\multirow{ 16 }{*}{}\\\cline{2-3}
\multirow{ 2 }{*}{}& $3$ & $2$ &\multirow{ 16 }{*}{}\\\cline{1-3}
\multirow{ 1 }{*}{ $6$ }& $2$ & $0, 3$ &\multirow{ 16 }{*}{}\\\cline{1-3}
\multirow{ 1 }{*}{ $7$ }& $2$ & $0, 3$ &\multirow{ 16 }{*}{}\\
\hline
\multirow{ 4 }{*}{ $0$ }& $4$ & $4, 5, 6, 7$ &\multirow{ 8 }{*}{\rotatebox{90}{$-86$}}\\\cline{2-3}
\multirow{ 4 }{*}{}& $5$ & $4, 5, 6, 7$ &\multirow{ 8 }{*}{}\\\cline{2-3}
\multirow{ 4 }{*}{}& $6$ & $4, 5, 6, 7$ &\multirow{ 8 }{*}{}\\\cline{2-3}
\multirow{ 4 }{*}{}& $7$ & $4, 5, 6, 7$ &\multirow{ 8 }{*}{}\\\cline{1-3}
\multirow{ 4 }{*}{ $3$ }& 4 & $4, 5, 6, 7$ &\multirow{ 8 }{*}{}\\\cline{2-3}
\multirow{ 4 }{*}{}& $5$ & $4, 5, 6, 7$ &\multirow{ 8 }{*}{}\\\cline{2-3}
\multirow{ 4 }{*}{}& $6$ & $4, 5, 6, 7$ &\multirow{ 8 }{*}{}\\\cline{2-3}
\multirow{ 4 }{*}{}& $7$ & $4, 5, 6, 7$ &\multirow{ 8 }{*}{}\\
\hline
&  & $\Downarrow$ & \\
\end{tabular}
\end{minipage}
\hspace{.3cm}
\begin{minipage}{.18\linewidth}
\vspace{-2cm}
\centering
\begin{tabular}{||c|c|c|c||}
&  & \multirow{2}{*}{$\Downarrow$} & \\
&  & \multirow{2}{*}{} & \\
\hline
\multirow{ 6 }{*}{ $4$ }& $0$ & $4, 5, 6, 7$ &\multirow{ 24 }{*}{\rotatebox{90}{$-86$}}\\\cline{2-3}
\multirow{ 6 }{*}{}& $3$ & $4, 5, 6, 7$ &\multirow{ 24 }{*}{}\\\cline{2-3}
\multirow{ 6 }{*}{}& $4$ & $0, 3$ &\multirow{ 24 }{*}{}\\\cline{2-3}
\multirow{ 6 }{*}{}& $5$ & $0, 3$ &\multirow{ 24 }{*}{}\\\cline{2-3}
\multirow{ 6 }{*}{}& $6$ & $0, 3$ &\multirow{ 24 }{*}{}\\\cline{2-3}
\multirow{ 6 }{*}{}& $7$ & $0, 3$ &\multirow{ 24 }{*}{}\\\cline{1-3}
\multirow{ 6 }{*}{ $5$ }& 0 & $4, 5, 6, 7$ &\multirow{ 24 }{*}{}\\\cline{2-3}
\multirow{ 6 }{*}{}& $3$ & $4, 5, 6, 7$ &\multirow{ 24 }{*}{}\\\cline{2-3}
\multirow{ 6 }{*}{}& $4$ & $0, 3$ &\multirow{ 24 }{*}{}\\\cline{2-3}
\multirow{ 6 }{*}{}& $5$ & $0, 3$ &\multirow{ 24 }{*}{}\\\cline{2-3}
\multirow{ 6 }{*}{}& $6$ & $0, 3$ &\multirow{ 24 }{*}{}\\\cline{2-3}
\multirow{ 6 }{*}{}& $7$ & $0, 3$ &\multirow{ 24 }{*}{}\\\cline{1-3}
\multirow{ 6 }{*}{ $6$ }& $0$ & $4, 5, 6, 7$ &\multirow{ 24 }{*}{}\\\cline{2-3}
\multirow{ 6 }{*}{}& $3$ & $4, 5, 6, 7$ &\multirow{ 24 }{*}{}\\\cline{2-3}
\multirow{ 6 }{*}{}& $4$ & $0, 3$ &\multirow{ 24 }{*}{}\\\cline{2-3}
\multirow{ 6 }{*}{}& $5$ & $0, 3$ &\multirow{ 24 }{*}{}\\\cline{2-3}
\multirow{ 6 }{*}{}& $6$ & $0, 3$ &\multirow{ 24 }{*}{}\\\cline{2-3}
\multirow{ 6 }{*}{}& $7$ & $0, 3$ &\multirow{ 24 }{*}{}\\\cline{1-3}
\multirow{ 6 }{*}{ $7$ }& $0$ & $4, 5, 6, 7$ &\multirow{ 24 }{*}{}\\\cline{2-3}
\multirow{ 6 }{*}{}& $3$ & $4, 5, 6, 7$ &\multirow{ 24 }{*}{}\\\cline{2-3}
\multirow{ 6 }{*}{}& $4$ & $0, 3$ &\multirow{ 24 }{*}{}\\\cline{2-3}
\multirow{ 6 }{*}{}& $5$ & $0, 3$ &\multirow{ 24 }{*}{}\\\cline{2-3}
\multirow{ 6 }{*}{}& $6$ & $0, 3$ &\multirow{ 24 }{*}{}\\\cline{2-3}
\multirow{ 6 }{*}{}& $7$ & $0, 3$ &\multirow{ 24 }{*}{}\\
\hline
\multirow{ 3 }{*}{ $0$ }& $2$ & $6, 7$ &\multirow{ 16 }{*}{ \rotatebox{90}{$-80$} }\\\cline{2-3}
\multirow{ 3 }{*}{}& $4$ & $2$ &\multirow{ 16 }{*}{}\\\cline{2-3}
\multirow{ 3 }{*}{}& $5$ & $2$ &\multirow{ 16 }{*}{}\\\cline{1-3}
\multirow{ 4 }{*}{ $2$ }& $0$ & $4, 5$ &\multirow{ 16 }{*}{}\\\cline{2-3}
\multirow{ 4 }{*}{}& $3$ & $4, 5$ &\multirow{ 16 }{*}{}\\\cline{2-3}
\multirow{ 4 }{*}{}& $6$ & $0, 3$ &\multirow{ 16 }{*}{}\\\cline{2-3}
\multirow{ 4 }{*}{}& $7$ & $0, 3$ &\multirow{ 16 }{*}{}\\\cline{1-3}
\multirow{ 3 }{*}{ $3$ }& $2$ & $6, 7$ &\multirow{ 16 }{*}{}\\\cline{2-3}
\multirow{ 3 }{*}{}& $4$ & $2$ &\multirow{ 16 }{*}{}\\\cline{2-3}
\multirow{ 3 }{*}{}& $5$ & $2$ &\multirow{ 16 }{*}{}\\\cline{1-3}
\multirow{ 1 }{*}{ $4$ }& $2$ & $0, 3$ &\multirow{ 16 }{*}{}\\\cline{1-3}
\multirow{ 1 }{*}{ $5$ }& $2$ & $0, 3$ &\multirow{ 16 }{*}{}\\\cline{1-3}
\multirow{ 2 }{*}{ $6$ }& $0$ & $2$ &\multirow{ 16 }{*}{}\\\cline{2-3}
\multirow{ 2 }{*}{}& $3$ & $2$ &\multirow{ 16 }{*}{}\\\cline{1-3}
\multirow{ 2 }{*}{ $7$ }& $0$ & $2$ &\multirow{ 16 }{*}{}\\\cline{2-3}
\multirow{ 2 }{*}{}& $3$ & $2$ &\multirow{ 16 }{*}{}\\
\hline
\multirow{ 2 }{*}{ $2$ }& $4$ & $6, 7$ &\multirow{ 8 }{*}{ \rotatebox{90}{$-76$} }\\\cline{2-3}
\multirow{ 2 }{*}{}& $5$ & $6, 7$ &\multirow{ 8 }{*}{}\\\cline{1-3}
\multirow{ 2 }{*}{ $4$ }& $6$ & $2$ &\multirow{ 8 }{*}{}\\\cline{2-3}
\multirow{ 2 }{*}{}& $7$ & $2$ &\multirow{ 8 }{*}{}\\\cline{1-3}
\multirow{ 2 }{*}{ $5$ }& $6$ & $2$ &\multirow{ 8 }{*}{}\\\cline{2-3}
\multirow{ 2 }{*}{}& $7$ & $2$ &\multirow{ 8 }{*}{}\\\cline{1-3}
\multirow{ 1 }{*}{ $6$ }& $2$ & $4, 5$ &\multirow{ 8 }{*}{}\\\cline{1-3}
\multirow{ 1 }{*}{ $7$ }& $2$ & $4, 5$ &\multirow{ 8 }{*}{}\\
\hline
\multirow{ 4 }{*}{ $2$ }& $4$ & $4, 5$ &\multirow{ 13 }{*}{ \rotatebox{90}{$-68$} }\\\cline{2-3}
\multirow{ 4 }{*}{}& $5$ & $4, 5$ &\multirow{ 13 }{*}{}\\\cline{2-3}
\multirow{ 4 }{*}{}& $6$ & $6, 7$ &\multirow{ 13 }{*}{}\\\cline{2-3}
\multirow{ 4 }{*}{}& $7$ & $6, 7$ &\multirow{ 13 }{*}{}\\\cline{1-3}
\multirow{ 3 }{*}{ $4$ }& $2$ & $4, 5$ &\multirow{ 13 }{*}{}\\\cline{2-3}
\multirow{ 3 }{*}{}& $4$ & $2$ &\multirow{ 13 }{*}{}\\\cline{2-3}
\multirow{ 3 }{*}{}& $5$ & $2$ &\multirow{ 13 }{*}{}\\\cline{1-3}
\multirow{ 3 }{*}{ $5$ }& $2$ & $4, 5$ &\multirow{ 13 }{*}{}\\\cline{2-3}
\multirow{ 3 }{*}{}& $4$ & $2$ &\multirow{ 13 }{*}{}\\\cline{2-3}
\multirow{ 3 }{*}{}& $5$ & $2$ &\multirow{ 13 }{*}{}\\\cline{1-3}
\multirow{ 3 }{*}{ $6$ }& $2$ & $6, 7$ &\multirow{ 13 }{*}{}\\\cline{2-3}
\multirow{ 3 }{*}{}& $6$ & $2$ &\multirow{ 13 }{*}{}\\\cline{2-3}
\multirow{ 3 }{*}{}& $7$ & $2$ &\multirow{ 13 }{*}{}\\\cline{1-3}
\hline
&  & $\Downarrow$ & \\
\end{tabular}
\end{minipage}
\hspace{.2cm}
\begin{minipage}{.15\linewidth}
\vspace{-2cm}
\centering
\begin{tabular}{||c|c|c|c||}
&  & \multirow{2}{*}{$\Downarrow$} & \\
&  & \multirow{2}{*}{} & \\
\hline
\multirow{ 3 }{*}{ $7$ }& $2$ & $6, 7$ &\multirow{ 3}{*}{\rotatebox{90}{$-68$}}\\\cline{2-3}
\multirow{ 3 }{*}{}& $6$ & $2$ &\multirow{ 3 }{*}{}\\\cline{2-3}
\multirow{ 3 }{*}{}& $7$ & $2$ &\multirow{ 3 }{*}{}\\
\hline
\multirow{ 6 }{*}{ $0$ }& $1$ & $2, 4, 5, 6, 7$ &\multirow{ 43 }{*}{ \rotatebox{90}{$-64$}}\\\cline{2-3}
\multirow{ 6 }{*}{}& $2$ & $1$ &\multirow{ 43 }{*}{}\\\cline{2-3}
\multirow{ 6 }{*}{}& $4$ & $1$ &\multirow{ 43 }{*}{}\\\cline{2-3}
\multirow{ 6 }{*}{}& $5$ & $1$ &\multirow{ 43 }{*}{}\\\cline{2-3}
\multirow{ 6 }{*}{}& $6$ & $1$ &\multirow{ 43 }{*}{}\\\cline{2-3}
\multirow{ 6 }{*}{}& $7$ & $1$ &\multirow{ 43 }{*}{}\\\cline{1-3}
\multirow{ 7 }{*}{ $1$ }& $0$ & $2, 4, 5, 6, 7$ &\multirow{ 43 }{*}{}\\\cline{2-3}
\multirow{ 7 }{*}{}& $2$ & $0, 3$ &\multirow{ 43 }{*}{}\\\cline{2-3}
\multirow{ 7 }{*}{}& $3$ & $2, 4, 5, 6, 7$ &\multirow{ 43 }{*}{}\\\cline{2-3}
\multirow{ 7 }{*}{}& $4$ & $0, 3$ &\multirow{ 43 }{*}{}\\\cline{2-3}
\multirow{ 7 }{*}{}& $5$ & $0, 3$ &\multirow{ 43 }{*}{}\\\cline{2-3}
\multirow{ 7 }{*}{}& $6$ & $0, 3$ &\multirow{ 43 }{*}{}\\\cline{2-3}
\multirow{ 7 }{*}{}& $7$ & $0, 3$ &\multirow{ 43 }{*}{}\\\cline{1-3}
\multirow{ 8 }{*}{ $2$ }& $0$ & $1$ &\multirow{ 43 }{*}{}\\\cline{2-3}
\multirow{ 8 }{*}{}& $1$ & $0, 3$ &\multirow{ 43 }{*}{}\\\cline{2-3}
\multirow{ 8 }{*}{}& $2$ & $4, 5, 6, 7$ &\multirow{ 43 }{*}{}\\\cline{2-3}
\multirow{ 8 }{*}{}& $3$ & $1$ &\multirow{ 43 }{*}{}\\\cline{2-3}
\multirow{ 8 }{*}{}& $4$ & $2$ &\multirow{ 43 }{*}{}\\\cline{2-3}
\multirow{ 8 }{*}{}& $5$ & $2$ &\multirow{ 43 }{*}{}\\\cline{2-3}
\multirow{ 8 }{*}{}& $6$ & $2$ &\multirow{ 43 }{*}{}\\\cline{2-3}
\multirow{ 8 }{*}{}& $7$ & $2$ &\multirow{ 43 }{*}{}\\\cline{1-3}
\multirow{ 6 }{*}{ $3$ }& $1$ & $2, 4, 5, 6, 7$ &\multirow{ 43 }{*}{}\\\cline{2-3}
\multirow{ 6 }{*}{}& $2$ & $1$ &\multirow{ 43 }{*}{}\\\cline{2-3}
\multirow{ 6 }{*}{}& $4$ & $1$ &\multirow{ 43 }{*}{}\\\cline{2-3}
\multirow{ 6 }{*}{}& $5$ & $1$ &\multirow{ 43 }{*}{}\\\cline{2-3}
\multirow{ 6 }{*}{}& $6$ & $1$ &\multirow{ 43 }{*}{}\\\cline{2-3}
\multirow{ 6 }{*}{}& $7$ & $1$ &\multirow{ 43 }{*}{}\\\cline{1-3}
\multirow{ 4 }{*}{ $4$ }& $0$ & $1$ &\multirow{ 43 }{*}{}\\\cline{2-3}
\multirow{ 4 }{*}{}& $1$ & $0, 3$ &\multirow{ 43 }{*}{}\\\cline{2-3}
\multirow{ 4 }{*}{}& $2$ & $2$ &\multirow{ 43 }{*}{}\\\cline{2-3}
\multirow{ 4 }{*}{}& $3$ & $1$ &\multirow{ 43 }{*}{}\\\cline{1-3}
\multirow{ 4 }{*}{ $5$ }& $0$ & $1$ &\multirow{ 43 }{*}{}\\\cline{2-3}
\multirow{ 4 }{*}{}& $1$ & $0, 3$ &\multirow{ 43 }{*}{}\\\cline{2-3}
\multirow{ 4 }{*}{}& $2$ & $2$ &\multirow{ 43 }{*}{}\\\cline{2-3}
\multirow{ 4 }{*}{}& $3$ & $1$ &\multirow{ 43 }{*}{}\\\cline{1-3}
\multirow{ 4 }{*}{ $6$ }& $0$ & $1$ &\multirow{ 43 }{*}{}\\\cline{2-3}
\multirow{ 4 }{*}{}& $1$ & $0, 3$ &\multirow{ 43 }{*}{}\\\cline{2-3}
\multirow{ 4 }{*}{}& $2$ & $2$ &\multirow{ 43 }{*}{}\\\cline{2-3}
\multirow{ 4 }{*}{}& $3$ & $1$ &\multirow{ 43 }{*}{}\\\cline{1-3}
\multirow{ 4 }{*}{ $7$ }& $0$ & $1$ &\multirow{ 43 }{*}{}\\\cline{2-3}
\multirow{ 4 }{*}{}& $1$ & $0, 3$ &\multirow{ 43 }{*}{}\\\cline{2-3}
\multirow{ 4 }{*}{}& $2$ & $2$ &\multirow{ 43 }{*}{}\\\cline{2-3}
\multirow{ 4 }{*}{}& $3$ & $1$ &\multirow{ 43 }{*}{}\\
\hline
\multirow{ 2 }{*}{$ 2$ }& $6$ & $4, 5$ &\multirow{ 8 }{*}{\rotatebox{90}{$-54$}}\\\cline{2-3}
\multirow{ 2 }{*}{}& $7$ & $4, 5$ &\multirow{ 8 }{*}{}\\\cline{1-3}
\multirow{ 1 }{*}{ $4$ }& $2$ & $6, 7$ &\multirow{ 8 }{*}{}\\\cline{1-3}
\multirow{ 1 }{*}{ $5$ }& $2$ & $6, 7$ &\multirow{ 8 }{*}{}\\\cline{1-3}
\multirow{ 2 }{*}{ $6$ }& $4$ & $2$ &\multirow{ 8 }{*}{}\\\cline{2-3}
\multirow{ 2 }{*}{}& $5$ & $2$ &\multirow{ 8 }{*}{}\\\cline{1-3}
\multirow{ 2 }{*}{ $7$ }& $4$ & $2$ &\multirow{ 8 }{*}{}\\\cline{2-3}
\multirow{ 2 }{*}{}& $5$ & $2$ &\multirow{ 8 }{*}{}\\
\hline
\multirow{ 1 }{*}{ $0$ }& $1$ & $1$ &\multirow{ 5 }{*}{\rotatebox{90}{$-48$}}\\\cline{1-3}
\multirow{ 3 }{*}{ $1$ }& $0$ & $1$ &\multirow{ 5 }{*}{}\\\cline{2-3}
\multirow{ 3 }{*}{}& $1$ & $0, 3$ &\multirow{ 5 }{*}{}\\\cline{2-3}
\multirow{ 3 }{*}{}& $3$ & $1$ &\multirow{ 5 }{*}{}\\\cline{1-3}
\multirow{ 1 }{*}{ $3$ }& $1$ & $1$ &\multirow{ 5 }{*}{}\\
\hline
\multirow{ 2 }{*}{ $4$ }& $4$ & $4, 5, 6, 7$ &\multirow{ 2 }{*}{\rotatebox{90}{$-47$}}\\\cline{2-3}
\multirow{ 2 }{*}{}& $5$ & $4, 5, 6, 7$ &\multirow{ 2 }{*}{}\\
\hline
&  & $\Downarrow$ & \\
\end{tabular}
\end{minipage}
\hspace{.2cm}
\begin{minipage}{.17\linewidth}
\vspace{-2cm}
\centering
\begin{tabular}{||c|c|c|c||}
&  & \multirow{2}{*}{$\Downarrow$} & \\
&  & \multirow{2}{*}{} & \\
\hline
\multirow{ 2 }{*}{$4$}& $6$ & $4, 5, 6, 7$ &\multirow{ 14 }{*}{\rotatebox{90}{$-47$}}\\\cline{2-3}
\multirow{ 2 }{*}{}& $7$ & $4, 5, 6, 7$ &\multirow{ 14 }{*}{}\\\cline{1-3}
\multirow{ 4 }{*}{ $5$ }& $4$ & $4, 5, 6, 7$ &\multirow{ 14 }{*}{}\\\cline{2-3}
\multirow{ 4 }{*}{}& $5$ & $4, 5, 6, 7$ &\multirow{ 14 }{*}{}\\\cline{2-3}
\multirow{ 4 }{*}{}& $6$ & $4, 5, 6, 7$ &\multirow{ 14 }{*}{}\\\cline{2-3}
\multirow{ 4 }{*}{}& $7$ & $4, 5, 6, 7$ &\multirow{ 14 }{*}{}\\\cline{1-3}
\multirow{ 4 }{*}{ $6$ }& $4$ & $4, 5, 6, 7$ &\multirow{ 14 }{*}{}\\\cline{2-3}
\multirow{ 4 }{*}{}& $5$ & $4, 5, 6, 7$ &\multirow{ 14 }{*}{}\\\cline{2-3}
\multirow{ 4 }{*}{}& $6$ & $4, 5, 6, 7$ &\multirow{ 14 }{*}{}\\\cline{2-3}
\multirow{ 4 }{*}{}& $7$ & $4, 5, 6, 7$ &\multirow{ 14 }{*}{}\\\cline{1-3}
\multirow{ 4 }{*}{ $7$ }& $4$ & $4, 5, 6, 7$ &\multirow{ 14 }{*}{}\\\cline{2-3}
\multirow{ 4 }{*}{}& $5$ & $4, 5, 6, 7$ &\multirow{ 14 }{*}{}\\\cline{2-3}
\multirow{ 4 }{*}{}& $6$ & $4, 5, 6, 7$ &\multirow{ 14 }{*}{}\\\cline{2-3}
\multirow{ 4 }{*}{}& $7$ & $4, 5, 6, 7$ &\multirow{ 14 }{*}{}\\
\hline
\multirow{ 3 }{*}{ $1$ }& $2$ & $4, 5$ &\multirow{ 10 }{*}{\rotatebox{90}{$-44$} }\\\cline{2-3}
\multirow{ 3 }{*}{}& $6$ & $2$ &\multirow{ 10 }{*}{}\\\cline{2-3}
\multirow{ 3 }{*}{}& $7$ & $2$ &\multirow{ 10 }{*}{}\\\cline{1-3}
\multirow{ 3 }{*}{ $2$ }& $1$ & $6, 7$ &\multirow{ 10 }{*}{}\\\cline{2-3}
\multirow{ 3 }{*}{}& $4$ & $1$ &\multirow{ 10 }{*}{}\\\cline{2-3}
\multirow{ 3 }{*}{}& $5$ & $1$ &\multirow{ 10 }{*}{}\\\cline{1-3}
\multirow{ 1 }{*}{ $4$ }& $1$ & $2$ &\multirow{ 10 }{*}{}\\\cline{1-3}
\multirow{ 1 }{*}{ $5$ }& $1$ & $2$ &\multirow{ 10 }{*}{}\\\cline{1-3}
\multirow{ 1 }{*}{ $6$ }& $2$ & $1$ &\multirow{ 10 }{*}{}\\\cline{1-3}
\multirow{ 1 }{*}{ $7$ }& $2$ & $1$ &\multirow{ 10 }{*}{}\\
\hline
\multirow{ 2 }{*}{ $1$ }& $6$ & $4, 5$ &\multirow{ 8 }{*}{\rotatebox{90}{$-40$}}\\\cline{2-3}
\multirow{ 2 }{*}{}& $7$ & $4, 5$ &\multirow{ 8 }{*}{}\\\cline{1-3}
\multirow{ 1 }{*}{ $4$ }& $1$ & $6, 7$ &\multirow{ 8 }{*}{}\\\cline{1-3}
\multirow{ 1 }{*}{$ 5$ }& $1$ & $6, 7$ &\multirow{ 8 }{*}{}\\\cline{1-3}
\multirow{ 2 }{*}{ $6$ }& $4$ & $1$ &\multirow{ 8 }{*}{}\\\cline{2-3}
\multirow{ 2 }{*}{}& $5$ & $1$ &\multirow{ 8 }{*}{}\\\cline{1-3}
\multirow{ 2 }{*}{ $7 $}& 4 & $1$ &\multirow{ 8 }{*}{}\\\cline{2-3}
\multirow{ 2 }{*}{}& $5$ & $1$ &\multirow{ 8 }{*}{}\\
\hline
\multirow{ 1 }{*}{ $1$ }& $2$ & $2$ &\multirow{ 3 }{*}{\rotatebox{90}{$-32$}}\\\cline{1-3}
\multirow{ 2 }{*}{ $2$ }& $1$ & $2$ &\multirow{ 3 }{*}{}\\\cline{2-3}
\multirow{ 2 }{*}{}& $2$ & $1$ &\multirow{ 3 }{*}{}\\
\hline
\multirow{ 3 }{*}{ $1$ }& $2$ & $6, 7$ &\multirow{ 10 }{*}{\rotatebox{90}{$-28$} }\\\cline{2-3}
\multirow{ 3 }{*}{}& $4$ & $2$ &\multirow{ 10 }{*}{}\\\cline{2-3}
\multirow{ 3 }{*}{}& $5$ & $2$ &\multirow{ 10 }{*}{}\\\cline{1-3}
\multirow{ 3 }{*}{ $2$ }& $1$ & $4, 5$ &\multirow{ 10 }{*}{}\\\cline{2-3}
\multirow{ 3 }{*}{}& $6$ & $1$ &\multirow{ 10 }{*}{}\\\cline{2-3}
\multirow{ 3 }{*}{}& $7$ & $1$ &\multirow{ 10 }{*}{}\\\cline{1-3}
\multirow{ 1 }{*}{ $4$ }& $2$ & $1$ &\multirow{ 10 }{*}{}\\\cline{1-3}
\multirow{ 1 }{*}{ $5$ }& $2$ & $1$ &\multirow{ 10 }{*}{}\\\cline{1-3}
\multirow{ 1 }{*}{ $6$ }& $1$ & $2$ &\multirow{ 10 }{*}{}\\\cline{1-3}
\multirow{ 1 }{*}{ $7$ }& $1$ & $2$ &\multirow{ 10 }{*}{}\\
\hline
\multirow{ 4 }{*}{ $1$ }& $4$ & $4, 5$ &\multirow{ 16 }{*}{\rotatebox{90}{$-26$}}\\\cline{2-3}
\multirow{ 4 }{*}{}& $5$ & $4, 5$ &\multirow{ 16 }{*}{}\\\cline{2-3}
\multirow{ 4 }{*}{}& $6$ & $6, 7$ &\multirow{ 16 }{*}{}\\\cline{2-3}
\multirow{ 4 }{*}{}& $7$ & $6, 7$ &\multirow{ 16 }{*}{}\\\cline{1-3}
\multirow{ 3 }{*}{ $4$ }& $1$ & $4, 5$ &\multirow{ 16 }{*}{}\\\cline{2-3}
\multirow{ 3 }{*}{}& $4$ & $1$ &\multirow{ 16 }{*}{}\\\cline{2-3}
\multirow{ 3 }{*}{}& $5$ & $1$ &\multirow{ 16 }{*}{}\\\cline{1-3}
\multirow{ 3 }{*}{ $5$ }& $1$ & $4, 5$ &\multirow{ 16 }{*}{}\\\cline{2-3}
\multirow{ 3 }{*}{}& $4$ & $1$ &\multirow{ 16 }{*}{}\\\cline{2-3}
\multirow{ 3 }{*}{}& $5$ & $1$ &\multirow{ 16 }{*}{}\\\cline{1-3}
\multirow{ 3 }{*}{ $6$ }& $1$ & $6, 7$ &\multirow{ 16 }{*}{}\\\cline{2-3}
\multirow{ 3 }{*}{}& $6$ & $1$ &\multirow{ 16 }{*}{}\\\cline{2-3}
\multirow{ 3 }{*}{}& $7$ & $1$ &\multirow{ 16 }{*}{}\\\cline{1-3}
\multirow{ 3 }{*}{ $7$ }& $1$ & $6, 7$ &\multirow{ 16 }{*}{}\\\cline{2-3}
\multirow{ 3 }{*}{}& $6$ & $1$ &\multirow{ 16 }{*}{}\\\cline{2-3}
\multirow{ 3 }{*}{}& $7$ & $1$ &\multirow{ 16 }{*}{}\\
\hline
&  & $\Downarrow$ & \\
\end{tabular}
\end{minipage}
\hspace{.2cm}
\vspace{1cm}
\begin{minipage}{.17\linewidth}
\vspace{2cm}
\centering
\begin{tabular}{||c|c|c|c||}
&  & \multirow{2}{*}{$\Downarrow$} & \\
&  & \multirow{2}{*}{} & \\
\hline
\multirow{ 2 }{*}{ $1$ }& $4$ & $6, 7$ &\multirow{ 8 }{*}{\rotatebox{90}{$-18$}}\\\cline{2-3}
\multirow{ 2 }{*}{}& $5$ & $6, 7$ &\multirow{ 8 }{*}{}\\\cline{1-3}
\multirow{ 2 }{*}{ $4$ }& $6$ & $1$ &\multirow{ 8 }{*}{}\\\cline{2-3}
\multirow{ 2 }{*}{}& $7$ & $1$ &\multirow{ 8 }{*}{}\\\cline{1-3}
\multirow{ 2 }{*}{ $5$ }& $6$ & $1$ &\multirow{ 8 }{*}{}\\\cline{2-3}
\multirow{ 2 }{*}{}& $7$ & $1$ &\multirow{ 8 }{*}{}\\\cline{1-3}
\multirow{ 1 }{*}{ $6$ }& $1$ & $4, 5$ &\multirow{ 8 }{*}{}\\\cline{1-3}
\multirow{ 1 }{*}{ $7$ }& $1$ & $4, 5$ &\multirow{ 8 }{*}{}\\
\hline
\multirow{ 6 }{*}{ $1$ }& $1$ & $2, 4, 5, 6, 7$ &\multirow{ 11 }{*}{\rotatebox{90}{$-8$}}\\\cline{2-3}
\multirow{ 6 }{*}{}& $2$ & $1$ &\multirow{ 11 }{*}{}\\\cline{2-3}
\multirow{ 6 }{*}{}& $4$ & $1$ &\multirow{ 11 }{*}{}\\\cline{2-3}
\multirow{ 6 }{*}{}& $5$ & $1$ &\multirow{ 11 }{*}{}\\\cline{2-3}
\multirow{ 6 }{*}{}& $6$ & $1$ &\multirow{ 11 }{*}{}\\\cline{2-3}
\multirow{ 6 }{*}{}& $7$ & $1$ &\multirow{ 11 }{*}{}\\\cline{1-3}
\multirow{ 1 }{*}{ $2$ }& $1$ & $1$ &\multirow{ 11 }{*}{}\\\cline{1-3}
\multirow{ 1 }{*}{ $4$ }& $1$ & $1$ &\multirow{ 11 }{*}{}\\\cline{1-3}
\multirow{ 1 }{*}{ $5$ }& $1$ & $1$ &\multirow{ 11 }{*}{}\\\cline{1-3}
\multirow{ 1 }{*}{ $6$ }& $1$ & $1$ &\multirow{ 11 }{*}{}\\\cline{1-3}
\multirow{ 1 }{*}{ $7$ }& $1$ & $1$ &\multirow{ 11 }{*}{}\\
\hline
\multirow{ 1 }{*}{$\mathbf{1}$}& $\mathbf{1}$ & $\mathbf{1}$ &\multirow{ 1 }{*}{$\mathbf{8}$}\\
\hline
\end{tabular}
\caption{Optimal Total Payoff (TP) are listed for all the $8^3$ encoding strategies $(\mathrm{e}^i,\mathrm{e}^j,\mathrm{e}^k)\equiv(i,j,k)$, for $i,j,k\in\{0,1,\cdots,7\}$. The optimal classical payoff turns out to be $\mathbf{8}$, achieved for the encoding $(\mathrm{e}^1,\mathrm{e}^1,\mathrm{e}^1)\equiv(1,1,1)$. }
\label{paylist}
\vspace{10cm}
\end{minipage}
\hspace{.2cm}
\end{scriptsize}
\end{table}

\newpage
\begin{acknowledgements}
SGN acknowledges support from the CSIR project 09/0575(15951)/2022-EMR-I. MA and MB acknowledge funding from the National Mission in Interdisciplinary Cyber-Physical systems from the Department of Science and Technology through the I-HUB Quantum Technology Foundation (Grant no: I-HUB/PDF/2021-22/008). MB acknowledges support through the research grant of INSPIRE Faculty fellowship from the Department of Science and Technology, Government of India. EPL acknowledges funding from the QuantERA II Programme that has received funding from the European Union’s Horizon 2020 research and innovation programme under Grant Agreement No 101017733 and the F.R.S-FNRS Pint-Multi programme under Grant Agreement R.8014.21, from the European Union’s Horizon Europe research and innovation programme under the project ``Quantum Security Networks Partnership'' (QSNP, grant agreement No 101114043), from the F.R.S-FNRS through the PDR T.0171.22, from the FWO and F.R.S.-FNRS under the Excellence of Science (EOS) programme project 40007526, from the FWO through the BeQuNet SBO project S008323N, from the Belgian Federal Science Policy through the contract RT/22/BE-QCI and the EU ``BE-QCI'' program. EPL is funded by the European Union. Views and opinions expressed are however those of the author only and do not necessarily reflect those of the European Union. The European Union cannot be held responsible for them. EPL is grateful to Stefano Pironio for discussions on numerical methods.
\end{acknowledgements}



\begin{thebibliography}{99}
\bibitem{self1} For a better exposition, consider the shared resource $\$_{S_1S_3R}\cup\$_{S_1S_5}$. While $\$_{S_1S_3R}$ allows $S_1,~S_3$, and $R$ to correlate their encoding-decoding strategies, $S_1$ and $S_5$ can correlate their encoding strategies using $\$_{S_1S_5}$. At this point, it should be noted that given $\$_{S_1S_3R}\cup\$_{S_1S_5}$ as resources, it is not be possible for $S_3$ and $S_5$ to employ a correlated strategy.
\end{thebibliography}
\end{document}